\documentclass[12pt]{article}
\usepackage[usenames]{color}
\usepackage{graphicx, subfigure}
\usepackage{amsmath, amsthm, amssymb}
\usepackage{amsfonts}
\usepackage{fullpage}
\usepackage{ifthen}
\usepackage{url}
\usepackage[sort&compress]{natbib}
\usepackage{multirow}
\usepackage{bm}
\usepackage{tikz-qtree}

\usepackage[linesnumbered,algoruled,boxed,lined,commentsnumbered]{algorithm2e}

\newtheorem{theorem}{Theorem}[section]

\makeatletter
\def\@biblabel#1{}
\makeatother

\theoremstyle{plain}

\theoremstyle{definition}
\newtheorem{example}{Example}

\theoremstyle{remark}

\newcommand{\ex}{{\sf Exp}}

\newcommand{\gam}{{\sf Gamma}}

\title{Stochastic behavior of an $n$-node blockchain under cyber attacks from multiple hackers with random re-setting times}
\author{Xiufeng Xu$^*$\qquad Liang Hong\footnote{Department of Mathematical Sciences, The University of Texas at Dallas, 800 West Campbell Road, Richardson, TX 75080, USA.}}

\date{\today}

\usepackage{tikz}
\usetikzlibrary{automata,arrows,positioning,calc}
\usepackage{listings}
\usepackage[labelformat=empty]{caption}
\usepackage{subfig}
\usepackage{caption}
\begin{document}

\maketitle

\begin{abstract}
This paper investigates the stochastic behavior of an $n$-node blockchain which is continuously monitored and faces non-stop cyber attacks from multiple hackers.  The blockchain will start being re-set once hacking is detected, forfeiting previous efforts of all hackers.  It is assumed the re-setting process takes a random amount of time.  Multiple independent hackers will keep attempting to hack into the blockchain until one of them succeeds. For arbitrary distributions of the hacking times, detecting times, and re-setting times, we derive the instantaneous functional probability, the limiting functional probability,  and the mean functional time of the blockchain.  Moreover, we establish that these quantities are increasing functions of the number of nodes,  formalizing the intuition that the more nodes a blockchain has the more secure it is.

\smallskip

\emph{Keywords and phrases:} blockchain; instantaneous functional probability; limiting functional probability; mean functional times; multiple hackers; random re-setting time.
\end{abstract}

\section{Introduction}
The Internet has become an inseparable part of our life.  Nowadays most organizations have an online platform for conducting business with their customers.  The convenience and efficiency provided by this online business mode come with a serious risk: an organization has become more vulnerable to cyber attacks that could cost them a fortune.  The devastating effect of a cyber attack is as fresh as in recent memory,  as evidenced by the Colonial Pipeline ransomware attack in 2021,  the service attack on many major US airports in 2022,  the cyber attack against DP World Australia Port in 2023,  and the Lurie Children's Hospital ransomware attack in 2024.  Each organization must protect its information and maintain its online business.  One solution is to back up all its data continuously.  Unfortunately, this idea is not practical for most organizations due to the astronomical amount of information they have to store.  A viable option is to store sensitive data on a blockchain.  A blockchain consists of several identical computers in a network; each is called a \emph{node} and stores identical data as other nodes do.  All nodes continuously verify the data according to the majority rule---any piece of data is considered valid only if it is consistent with its counterpart on at least half of the nodes.  While a hacker can steal the data stored on a node of a blockchain,  the stolen data is often skillfully encrypted and hence is almost useless to the hacker due to the public key encryption technology (e.g.,  Diffie and Hellman 1976, Rivest et al. 1978, Goldwasser and Micali 1982, Goldwasser et al. 1988.  Therefore, most hackers would either wreak havoc on the organization by altering the data stored on the blockchain or making a ransom demand by locking the data on the blockchain.  The former type of cyber attack, called \emph{cyber destructive attack},  requires a hacker to hack into at least half of the nodes; the latter type, called \emph{cyber ransom attack}, requires the hacker to hack into all nodes.  In view of this,  blockchains do not provide cyber-attack-proof protection, but they reduce an organization's likelihood of incurring cyber losses.  For a comprehensive review of elements of blockchains and their applications, see Tama et al. (2017),  Zheng et al. (2018),  Casino et al. (2019), Kumar et al. (2020), and references therein. 

The security issues of a blockchain have received significant attention from the scientific community; see,  for example,  Corbet et al. (2020),  Hussain et al. (2022), Khanum and Mustafa (2022),  Khan and Salah (2017), Li et al. (2071), Meng et al. (2018), and Tsuchiya et al.  (2021).  However,  there are only a few papers that study the operations-research-theoretic aspect of a blockchain; see,  for instance,  Choi et al. (2020), Melo et al. (2021),  Xu and Hong (2014),  and Zhang et al. (2020). 
In particular,  Xu and Hong (2024) investigate the stochastic behavior of an $n$-node blockchain under cyber attacks from a single hacker.  This article is a further study along this line.  For arbitrary distributions of the detecting times,  hacking times, and re-setting times,  it derives several key quantities of an $n$-node blockchain under cyber attacks from multiple hackers when it takes a random amount of time to re-set the blockchain. 

The remainder of the article is organized as follows.  Section~2 describes the problem under consideration and establishes some notational conventions.  Section~3 investigates the stochastic behavior of an $n$-node blockchain in our setting and gives formulas for the instantaneous functional probability, limiting functional probability,  and mean functional time.  It also establishes that these quantities are increasing functions of the number of nodes. This substantiates the intuition that the more nodes a blockchain has the harder it is to hack into it.  Section~4 provides several numerical examples.  Many technical details of these examples are delegated to the Appendix.  Section~5 concludes the article with some remarks. Python code for all examples is available at \url{https://github.com/xuxiufeng/Blockchain_Simulation_Multiple_Hackers_with_Resetting_Times}.

\section{Model setup and notation}

We consider a blockchain of $n$ identical nodes where $n \geq 2$.   The blockchain is under continuous monitoring (i.e., 24-hour monitoring) of the IT department of an organization for potential cyber attacks.  
To perform a cyber destructive attack,  a hacker must hack into at least $m=\lfloor n/2\rfloor+1$ nodes.  We will focus on cyber destructive attacks only since the case of a cyber ransom attack is similar except $m$ is replaced by $n$. 
We assume there are $k$ independent hackers where $k\geq 1$, and each of them continuously attacks the blockchain until successfully hacking into the blockchain.  Once a hacker hacks into the blockchain, he will immediately start to alter the data on it.  When that happens, we say the blockchain is \emph{dysfunctional}.  Hence, the blockchain will be functional until at least one of the $k$ hackers successfully hacks into $m$ nodes. The IT department will start re-setting the blockchain once they detect that the blockchain is under attack.  Resetting takes a random amount of time to finish. After that,  the previous effort of each hacker is forfeited. 

Without lost of generality, we label the $n$ identical nodes as node $1$, node $2$, $\dots$,  and node $n$.  For $i=1,2,\dots,n$ and $j=1,2,\dots,k$,  let $X_i^j$ be the time the $j$-th hacker needs to hack into node $i$.  For $j=1, \ldots, k$, we assume that $X_1^j, \ldots, X_n^j$  are independent and identically distributed (iid) according to a (cumulative) distribution function of $F_{X^j_1}$ with a probability density function $f_{X^j_1}(x)$.  Let $Y$ be the time the IT department takes to detect any cyber attack. We assume that $Y$ has a cumulative function $F_Y$ with a probability density function $f_Y(y)$.   Let $W$ be the time it takes for the IT department to re-set the blockchain. We assume that $W$ has a distribution function $F_W$ with a probability density function $f_W(w)$.  Note that $X_i^j,Y,$ and $W$ are mutually independent for $j=1,\dots, k$. We call each $X_i$ a $hacking$ $time$, $Y$ a $detecting$ $time$, and $W$ a re-set time. For the $jth$ hacker, if $\sum_{i=1}^m X_i^j<Y$,  then the $j$-th hacker can change the data on the blockchain. If $\sum_{i=1}^m X_i^j\geq Y$,  then the $jth$ hacker is detected by the IT department, and the blockchain will start being re-set immediately.  It  takes $W$ amount of time to complete the re-setting process. Once the re-setting is finished,  all the $k$ hackers will conduct a new round of cyber attacks from the node $1$.

The stochastic behavior of the blockchain can be described as follows.  At time $t=0$, the blockchain starts to operate and all hackers start to work on hacking into the blockchain.  As all nodes are identical, we may assume that each hacker attacks the $n$ nodes in ascending order. That is, each of them will first attack node~1. If $Y\leq \min_{1\leq j\leq k}X^j_1$,  the IT department will detect the cyber attack before any hacker can hack into node~1. In that case,  the IT department will immediately start re-setting the blockchain, rendering all previous efforts of each hacker in vain. The re-setting process takes $W$ amount of time after which the blockchain operates as if afresh and all hackers start to attack node~1 again. On the other hand,  if $Y>\min_{1\leq j\leq k}X^j_1$, that is, at least one of the $k$ hackers jas successfully bypassed the firewall of node~1, he will immediately move onto attacking node~2.  This process continues until 
$Z_m=\min_{j=1,\dots,k}\bigg(\sum_{i=1}^m X_i^j\bigg)<Y$, i.e.,  one of the $k$ hackers succeeds in hacking into $m$ nodes. 

For given $m$ and $k$, we define the $\emph{instantaneous functional probability}$ of the blockchain to be
\begin{equation}
\label{eq:instantaneous}
P_{mk}(t)=P\{\text{the blockchain has not been hacked at time $t$}\}.
\end{equation}
As $t\rightarrow\infty$,  $P_{mk}(t)$ converges to the \emph{limiting functional probability}  of the blockchain which is defined as
\begin{equation}
\label{eq:limit}
P_{mk}(\infty)=P\{\text{the  blockchain will never be hacked}\}.
\end{equation}

\section{Stochastic behavior of the proposed blockchain model}

\subsection{Limiting functional probability}
To derive the limiting functional probability,  we define a $cycle$ to be the period from the moment the blockchain starts afresh or has just been re-set to the next moment it starts being re-set. Let $N_1$ be the total number of cycles before one of the $k$ hackers successfully hacks into the blockchain, i.e., $Z_m \leq Y$.  Then $N_1$ is a geometric random variable\footnote{Here we interpret a geometric random variable with parameter $p$ to be the number of failures until the first success among a sequence of independent Bernoulli trails with success probability $p$.} with  parameter $1-p_{mk}$, where
\begin{eqnarray}
\label{eq:detectprob}
	p_{mk}&=&1-P\bigg\{\sum_{i=1}^m X_i^1>Y, \sum_{i=1}^m X_i^2>Y,\dots, \sum_{i=1}^m X_i^k>Y\bigg \} \nonumber\\
	&=&1-P\bigg\{\min_{j=1,\dots,k}\bigg(\sum_{i=1}^m X_i^j\bigg)>Y\bigg \} \nonumber \\
	&=&1-P(Z_m>Y)
	=P(Z_m\leq Y)  \nonumber \\
	&=&\int_0^{\infty} F_{Z_m}(s) dF_{Y}(s),
\end{eqnarray}
where $F_{Z_m}$ is the distribution function of $Z_m$.  The continuous-time process seems to be mathematically intractable.  However,  if we focus only on those moments the blockchain starts afresh,  has just been re-set, or has just been hacked,  then we can identify an embedded three-state discrete-time Markov chain; see the following diagram.

\begin{center}
	\begin{tikzpicture}[->, >=stealth', auto, semithick, node distance=3cm]
	\tikzstyle{every state}=[fill=white,draw=black,thick,text=black,scale=2]
	\node[state]    (A)                     {$1$};     
	\node[state]    (B)[right of=A]   {$2$};
	\node[state]     (C)[left of=A] {$3$};
	\path
	(A) edge[bend left,above]			node{$1-p_{mk}$}(C)
	       edge[bend left,above]	node{$p_{mk}$}	(B)
	(B) edge[loop above]          node{$1$}(B)
	(C) edge[bend left,above]           node{$1$}(A);   
	\end{tikzpicture}          
\end{center}
\begin{center}
	Figure 1: Three-state semi-Markov process
\end{center}
In State $1$,  the blockchain is functional, i.e., none of the $k$ hackers has hacked into $m$ nodes.  In State $3$, the blockchain is dysfunctional (i.e., it has been hacked).  In State $2$,  the blockchain is being re-set.  State $2$ is an absorbing state.  The transition matrix of this Markov chain is given by (\ref{eq:transition}). 
\begin{equation}
\label{eq:transition}
	\begin{bmatrix}
0 & p_{mk} & 1-p_{mk}\\
0 & 1 & 0\\
1 & 0 & 0
\end{bmatrix}
\end{equation}
Condition on the outcome of the first cycle to obtain
\[P_m(\infty)=(1-p_{mk})\times 1\times P_{mk}(\infty)+p_{mk}\times 0.\]
It follows that $P_m(\infty)=0$.  Intuitively,  this says that if $p_{mk}>0$ then the blockchain  will eventually be hacked.


\subsection{Mean functional time}
To derive the mean functional time of blockchain, we identify a renewal process $\{N_2(t)\}_{t\geq 0}$ as follows:
\[N_2(t)= \max\left\{k: \sum_{i=1}^k (Y_i+W_i)\leq t \right\},\]
where $Y_i\overset{d}{=}Y\mid Y<Z_m$, i.e., $Y_i$ and the conditional random variable $Y\mid Y<Z_m$ have the same distribution. 
We put
\begin{equation}
\label{eq:totaltime}
T_m=\sum_{i=1}^{N_1} (Y_i +W_i)+\left(Z_m  \middle|\ Z_m<Y\right).
\end{equation}
Then $T_m$ is the functional time of the blockchain.
Since $N_1$ is independent of $Y_{n+1}, Y_{n+2}, \ldots$,  $N_1$ is a stopping time for the sequence
$\{Y_i\}_{i\geq 1}$.  It follows from Wald's Identity that
\begin{eqnarray}
\label{eq:meantime}
E[T_m] &=& E[N_1] (E[Y_1]+E[W_1])+E \left[ Z_m \middle|\ Z_m<Y\right ] \nonumber \\
     &=& E[N_1]\bigg\{ E\left [Y\middle|\ Y<Z_m \right]+E[W_1]\bigg\}+E \left[ Z_m \middle|\ Z_m<Y\right ].
\end{eqnarray}

\subsection{Instantaneous functional probability}
To obtain the instantaneous functional probability, we will establish a renewal-type equation. To this end, we 
define  $S_t$ to be the last moment the blockchain is re-set before time $t$. That is, 
\[S_t=\sum_{i=1}^{N_2(t)} Y_i+W_i.\]
Let $F_{S_t}$ be the distribution function of $S_t$.  Also, we use $G(s)=E[{N_2(s)}]$ to denote the renewal function of $\{N_2(t)\}_{t\geq 0}$. Then
	\[G(s)=\sum_{k=1} ^\infty P\left(\sum_{i=1}^k Y_i+W_i \leq s\right ).\]
Moreover,  we have 
\begin{eqnarray*}
	P(S_t=0)&=&P(Y_1+W_1>t)\\
	dF_{S_t}(s)&=&P(Y_1+W_1>t-s)dG(s), \quad \text{$0<s<\infty$}.
\end{eqnarray*}
To obtain $P_m(t)$, we need to derive the probability that the blockchain is functional at $t$ and the probability that the blockchain is re-setting at $t$.  For the former,  we have 
\begin{eqnarray*}
	&&P\{\text{the blockchain is functional at t}\}\\
=&&P\{\text{the blockchain is functional at $t$}\mid S_t=0\}P\{S_t=0\}\\
 &&+\int_0^tP\{\text{the blockchain is functional at $t$}\mid S_t=s\}dF_{S_t}(s)\\
 =&&P\left\{t<Z_m \wedge Y\right\}+\int_0^tP\left\{t-s<Z_m\wedge Y\right\}dG(s),
\end{eqnarray*}
where $dF_{S_t}(s)=P\{t-s<Y_1+W_1\}dG(s)$ and $G(s)=\sum_{k=1} ^\infty P\left(\sum_{i=1}^k Y_i+W_i \leq s\right )$. 
For the latter, we have
\begin{eqnarray*}
	&&P\{\text{the blockchain is re-setting at t}\}\\
=&&P\{\text{the blockchain is re-setting at $t$}\mid S_t=0\}P\{S_t=0\}\\
 &&+\int_0^tP\{\text{the blockchain is re-setting at $t$}\mid S_t=s\}dF_{S_t}(s)\\
 =&&P\left\{Y\leq t<Y+W_1,Y\leq Z_m \right\}\\
 &&+\int_0^tP\left\{Y\leq t-s<Y+W_1 ,Y\leq Z_m \right\}dG(s).\\
\end{eqnarray*} 
Therefore, 
\begin{eqnarray*}
\label{eq:instprob}
P_{mk}(t) 
&=& P\{\text{the blockchain has not been hacked at $t$}\mid S_t=0\}P\{S_t=0\}\nonumber\\
     & & +\int_0^t P\{ \text{the blockchain has not been hacked at $t$}\mid S_t=s\}dF_{S_t}(s)\nonumber\\
\end{eqnarray*}
\begin{eqnarray}
     &=&P\{\text{the blockchain is functional at $t$}\mid S_t=0\}P\{S_t=0\}\nonumber\\
     &&+P\{\text{the blockchain is re-setting at $t$}\mid S_t=0\}P\{S_t=0\}\nonumber\\
     &&+\int_0^tP\{\text{the blockchain is functional at $t$}\mid S_t=s\}dF_{S_t}(s)\nonumber\\
     &&+\int_0^tP\{\text{the blockchain is re-setting at $t$}\mid S_t=s\}dF_{S_t}(s)\nonumber\\
     &=&P\left\{t<Z_m \wedge Y\right\}+P\left\{Y\leq t<Y+W_1 ,Y\leq Z_m \right\}\nonumber\\
     &&+\int_0^tP\left\{t-s<Z_m \wedge Y\right\}dG(s)\nonumber\\
     &&+\int_0^tP\left\{Y\leq t-s<Y+W_1 ,Y\leq Z_m\right\}dG(s).
\end{eqnarray}



\subsection{Monotonicity of $P_m(t)$ and $E[T_m]$}
According to our intuition, the more nodes a blockchain has, the harder it is to hack into it.  Formally, this amounts to saying that  $P_m(t)$,  $P_m(\infty)$, and $E[T_m]$ all are increasing functions of $m$.  This is indeed the case, as the next theorem shows.
\begin{theorem}
\label{theorem3.1}
$P_{mk}(t)$,  $P_{mk}(\infty)$, and $E[T_m]$ are all increasing functions of  $m$, where $m=\lfloor n/2\rfloor+1$. 
\end{theorem}
\begin{proof}
First, the definition of $Z_m$ implies 
\begin{equation*}
P\{t<Y \wedge Z_{m+1} \}> P\{t<Y \wedge Z_m\},
\end{equation*}
and
\begin{equation*}
P\{t-s<Y \wedge Z_{m+1} \}> P\{t-s<Y \wedge Z_m\}.
\end{equation*}
Similarly, we have
\begin{equation*}
P\left\{Y\leq t<Y+W_1 ,Y\leq\ Z_{m+1}\right\} > P\left\{Y\leq t<Y+W_1 ,Y\leq Z_m\right\}	
\end{equation*}
and
\begin{equation*}
P\left\{Y\leq t-s<Y+W_1 ,Y\leq Z_{m+1}\right\} > P\left\{Y\leq t-s<Y+W_1 ,Y\leq Z_m\right\}.
\end{equation*}
Then it follows from (\ref{eq:instprob}) that 
\begin{equation}
\label{eq:instantaneous3}
 	P_{m+1,k}(t)> P_{mk}(t), \quad \text{for all positive integer $m$ and $t$},
\end{equation}
which also implies
\begin{equation}
 	P_{m+1,k}(\infty)=\lim_{t\rightarrow \infty}P_{m+1,k}(t)>\lim_{t\rightarrow \infty}P_{mk}(t)=P_{mk}(\infty).
 \end{equation}
To show $E[T_m]$ is increasing in $m$,  we write it as
\begin{eqnarray*}
E[T_m] &=& \frac{1-p_{mk}}{p_{mk}}\times \left\{\frac{E[Y1_{\{Y\leq Z_m \}}]}{1-p_{mk}}+E[W_1]\right\}+\frac{E[Z_m1_{\{Y>Z_m\}}]}{p_{mk}}\\
&=& \frac{E[Y1_{\{Y\leq Z_m\}}]}{p_{mk}}+\frac{E[Z_m1_{\{Y>Z_m\}}]}{p_{mk}}+\frac{1-p_{mk}}{p_{mk}}\times E[W_1].
\end{eqnarray*}
Thus,
\begin{eqnarray*}
	&&E[T_{m+1}]-E[T_m]\\
	&=&  (p_{mk}p_{m+1,k})^{-1} \bigg \{ p_{mk} E[Y1_{\{Y\leq Z_{m+1} \}}]-p_{m+1,k}E[Y1_{\{Y\leq Z_m \}}] \\
& &+p_{mk} E\bigg[Z_{m+1}1_{\{Y>Z_{m+1}\}}\bigg]-p_{m+1,k}E\bigg[Z_m1_{\{Y>Z_m\}}\bigg]\\
&&+(p_{mk}-p_{m+1,k})E[W_1] \bigg\}.\\
\end{eqnarray*}
The first four terms in the brackets can be written as 
\begin{eqnarray*}
& & p_{mk} E[Y1_{\{Y<Z_m\}}]-p_{m+1,k}E[Y1_{\{Y<Z_m \}}]+p_{mk} E[Y1_{\{Z_m<Y<Z_{m+1}\}}] \\
& & + p_{mk} E[(Z_{m+1}-Z_m)1_{\{Y>Z_{m+1}\}}]+ p_{mk} E[Z_m1_{\{Y>Z_{m+1}\}}]\\
& &-p_{m+1,k} E[Z_m1_{\{Y>Z_{m+1}\}}] 
 -p_{m+1,k} E[Z_m1_{\{Z_{m+1}>Y>Z_m\}}]\\
&=& (p_{mk}-p_{m+1,k})E[Y1_{\{Y<Z_m \}}]+(p_{mk}-p_{m+1,k})E[Z_m1_{\{Y>Z_{m+1}\}}]\\
& &  +p_{mk} E[Y1_{\{Z_{m+1}>Y>Z_m\}}]\ -p_{m+1,k} 
E[Z_m1_{\{Z_{m+1}>Y>Z_m\}}] \\
& & +p_{mk} E[(Z_{m+1}-Z_m)1_{\{Y>Z_{m+1}\}}] \\
&\geq& (p_{mk}-p_{m+1,k})E[Y1_{\{Y<Z_m \}}]+(p_{mk}-p_{m+1,k})E[Z_m1_{\{Y>Z_{m+1}\}}]\\
& &   +p_{mk} E[Z_m1_{\{Z_{m+1}>Y>Z_m\}}]\ -p_{m+1,k} 
E[Z_m1_{\{Z_{m+1}>Y>Z_m\}}]\\
& & +p_{mk} E[(Z_{m+1}-Z_m)1_{\{Y>Z_{m+1}\}}]\\
&=& (p_{mk}-p_{m+1,k})E[Y1_{\{Y<Z_m \}}]+(p_{mk}-p_{m+1,k})E[Z_m1_{\{Y>Z_{m+1}\}}]\\
& &   +(p_{mk}-p_{m+1,k}) E[Z_m1_{\{Z_{m+1}>Y>Z_m\}}] +p_{mk} E[(Z_{m+1}-Z_m)1_{\{Y>Z_{m+1}\}}].
\end{eqnarray*}
(\ref{eq:detectprob}) implies that $p_{mk}$ is decreasing in $m$. It is also clear that $Z_{m+1}-Z_m > 0$ almost surely. Hence, $ E[T_{m+1}]-E[T_m] >0$.
\end{proof}


\begin{theorem}
\label{theorem3.2}
$\lim_{m\rightarrow\infty}P_m(t)=1$ and $\lim_{m\rightarrow\infty}E[T_m]=\infty$, where $m=\lfloor n/2\rfloor+1$. 
\end{theorem}
\begin{proof}
It is clear from (\ref{eq:detectprob}) that $\lim_{m \rightarrow \infty}p_{mk}=0$.  This is equivalent to saying that  the embedded discrete-time Markov chain will never be  State $2$ almost surely in the limit.  Therefore, $\lim_{m \rightarrow \infty}P_m(t)=1$.  

Since $\lim_{m \rightarrow \infty}p_{mk}=0$, $E[N_1]=1/p_{mk}-1\rightarrow \infty$ as $m\rightarrow\infty$.  It follows from (\ref{eq:meantime}) that 
\begin{eqnarray*}
\lim_{m \rightarrow \infty}E[T_m ]
&\geq&	\lim_{m \rightarrow \infty}E[N_1] E\left [Y\middle|\ Y<Z_m\right]\\
&=& \lim_{m \rightarrow \infty}E[N_1] \lim_{m \rightarrow \infty}\left[ \frac{E[Y1_{\{Y<Z_m\}} ]}{P\left\{ Y<Z_m \right\}}\right] \\
&=&\lim_{m \rightarrow \infty}E[N_1] E[Y]\\
&=& \infty.
\end{eqnarray*}
\end{proof}

Intuitively, the more hackers out there, the more likely the blockchain will be hacker.  The next theorem confirms this intuition.

\begin{theorem}
\label{theorem3.4}
$P_{mk}(t)$,  $P_{mk}(\infty)$, and $E[T_{mk}]$ are all  increasing functions of  $k$, where $k$ is a positive integer. 
\end{theorem}
\begin{proof}
For $k_1>k_2$, we have
	\begin{eqnarray*}
		Z_{mk_2}=\min_{j=1,\dots,k_2}	\sum_{i=1}^m X_i^j \geq \min_{j=1,\dots,k_1}	\sum_{i=1}^m X_i^j=Z_{mk_1},
	\end{eqnarray*}
and hence
\begin{eqnarray*}
	Z_{mk_2}\wedge Y \geq Z_{mk_1} \wedge Y,
\end{eqnarray*}
Then (\ref{eq:instprob}) implies
\begin{eqnarray*}
	P_{mk_1}(t)-P_{mk_2}(t) &=& P\{t<Z_{mk_1}\wedge Y \}- P\{t<Z_{mk_2}\wedge Y \}\\
	&&+\int_0^t P\{Y \wedge Z_{mk_1}>t-s \}-\{Y \wedge Z_{mk_2}>t-s \} dG(s)\\
	&&+\int_0^tP\left\{Y\leq t-s<Y+W_1 ,Y\leq Z_{mk_1}\right\}\\
	&&-P\left\{Y\leq t-s<Y+W_1 ,Y\leq Z_{mk_2}\right\}dG(s)\\
	&\leq&0.
\end{eqnarray*}
This shows that $P_{mk}(t)$ is an increasing function of $k$.

Next, (\ref{eq:meantime}) implies
\begin{eqnarray*}
E[T_{mk}] 
&=& \frac{E[Y1_{\{Y\leq Z_{mk}\}}]}{p_{mk}}+\frac{E[Z_{mk}1_{\{Y>Z_{mk}\}}]}{p_{mk}}+\frac{1-p_{mk}}{p_{mk}}\times E[W_1].
\end{eqnarray*}
Thus, 
\begin{eqnarray*}
 E[T_{mk_1}]-E[T_{mk_2}] 
&=&  (p_{mk_2}p_{mk_1})^{-1} \bigg \{ p_{mk_2} E[Y1_{\{Y<Z_{mk_1} \}}]-p_{mk_1}E[Y1_{\{Y<Z_{mk_2} \}}] \\
& &+p_{mk_2} E[Z_{mk_1}1_{\{Y>Z_{mk_1}\}}]-p_{mk_1}E[Z_{mk_2}1_{\{Y>Z_{mk_2}\}}]+(p_{mk_2}-p_{mk_1})E[W_1] \bigg\}.
\end{eqnarray*}
The first four terms in the brackets can be written as 
\begin{eqnarray*}
& & p_{mk_2} E[Y1_{\{Y<Z_{mk_1}\}}]-p_{mk_1}E[Y1_{\{Y<Z_{mk_1} \}}]-p_{mk_1} E[Y1_{\{Z_{mk_1}<Y<Z_{mk_2}\}}] \\
& & + p_{mk_2} E[(Z_{mk_1}-Z_{mk_2})1_{\{Y>Z_{mk_1}\}}]+ p_{mk_2} E[Z_{mk_2}1_{\{Y>Z_{mk_1}\}}]\\
& &-p_{mk_1} E[Z_{mk_2}1_{\{Y>Z_{mk_1}\}}] 
 +p_{mk_1} E[Z_{mk_2}1_{\{Z_{mk_2}>Y>Z_{mk_1}\}}]\\
&=& (p_{mk_2}-p_{mk_1})E[Y1_{\{Y<Z_{mk_1} \}}]+(p_{mk_2}-p_{mk_1})E[Z_{mk_2}1_{\{Y>Z_{mk_1}\}}]\\
& &  +p_{mk_1} E[Z_{mk_2}1_{\{Z_{mk_2}>Y>Z_{mk_1}\}}]\ -p_{mk_1} E[Y1_{\{Z_{mk_1}<Y<Z_{mk_2}\}}] \\
& & +p_{mk_2} E[(Z_{mk_1}-Z_{mk_2})1_{\{Y>Z_{mk_1}\}}] \\
&\leq&(p_{mk_2}-p_{mk_1})E[Y1_{\{Y<Z_{mk_1} \}}]+(p_{mk_2}-p_{mk_1})E[Z_{mk_2}1_{\{Y>Z_{mk_1}\}}]\\
& &  +p_{mk_1} E[Z_{mk_2}1_{\{Z_{mk_2}>Y>Z_{mk_1}\}}]\ -p_{mk_1} E[Z_{mk_1}1_{\{Z_{mk_1}<Y<Z_{mk_2}\}}] \\
& & +p_{mk_2} E[(Z_{mk_1}-Z_{mk_2})1_{\{Y>Z_{mk_1}\}}] \\
&=&(p_{mk_2}-p_{mk_1})E[Y1_{\{Y<Z_{mk_1} \}}]+(p_{mk_2}-p_{mk_1})E[Z_{mk_2}1_{\{Y>Z_{mk_1}\}}]\\
& &  +p_{mk_1} E[(Z_{mk_2}-Z_{mk_1})1_{\{Z_{mk_2}>Y>Z_{mk_1}\}}]\\
& & +p_{mk_2} E[(Z_{mk_1}-Z_{mk_2})1_{\{Y>Z_{mk_1}\}}] \\
&\leq&(p_{mk_2}-p_{mk_1})E[Y1_{\{Y<Z_{mk_1} \}}]+(p_{mk_2}-p_{mk_1})E[Z_{mk_2}1_{\{Y>Z_{mk_1}\}}]\\
& &  +p_{mk_1} E[(Z_{mk_2}-Z_{mk_1})1_{\{Y>Z_{mk_1}\}}]\\
& & +p_{mk_2} E[(Z_{mk_1}-Z_{mk_2})1_{\{Y>Z_{mk_1}\}}] \\
&=&(p_{mk_2}-p_{mk_1})E[Y1_{\{Y<Z_{mk_1} \}}]+(p_{mk_2}-p_{mk_1})E[Z_{mk_2}1_{\{Y>Z_{mk_1}\}}]\\
& &  +(p_{mk_2} -p_{mk_1})E[(Z_{mk_1}-Z_{mk_2})1_{\{Y>Z_{mk_1}\}}]\\
&=&(p_{mk_2}-p_{mk_1})E[Y1_{\{Y<Z_{mk_1} \}}]+(p_{mk_2}-p_{mk_1})E[Z_{mk_1}1_{\{Y>Z_{mk_1}\}}]\\
&\leq& 0,
\end{eqnarray*}
Besides, $(p_{mk_2}-p_{mk_1})E[W_1]\leq 0$. Therefore, $E[T_{mk}]$ is an increasing function of $k$.
\end{proof}

\subsection{Cost-benefit analysis}
Suppose $C_1(m)$ is the cost per unit of time for re-setting the blockchain, $C_2(m)$ is the cost per unit of time for running the blockchain, and $R(m)$ is the revenue earned per unit of time by running the blockchain. Then the total net revenue (i.e., the total profit), denoted as $TNR$, can be written as 
\begin{equation}
	TNR=\left\{\sum_{i=1}^{N_1} Y_i +\left(Z_m  \middle|\ Z_m<Y\right)\right\}(R(m)-C_2(m))-C_1(m)\sum_{i=1}^{N_1} W_i.
	\end{equation}
It follows from Wald's Identity that  the total expected net revenue is given by the following equation:
\begin{equation}
\label{eq:totalrevenue}
E_m[TNR]=\left\{E[N_1]E\left [Y\middle|\ Y<Z_m \right]+E \left[ Z_m \middle|\ Z_m<Y\right ]\right\}(R(m)-C_2(m))-E[N_1]E[W_1]C_1(m).
\end{equation}
Therefore, the expected net revenue per unit of time, denoted as $E[NR]$, can be written as 
\begin{equation}
\label{eq:netrevenue}
E_m[NR]=\frac{E[TNR]}{E[T_m]}.
\end{equation}

\section{Examples}

In this section, we provide three numerical examples. Derivations of all formulas are provided in the Appendix.
\begin{example}[Exponential hacking, exponential re-set, and exponential detecting times]
\label{example1}

Suppose $X_i^j\overset{iid}{\sim} \ex(\lambda_j)$ for $i=1,\dots,n$ and $j=1,\dots,k$, $Y\sim \ex(\delta)$, and $W\sim \ex(\eta)$ such that $X_1^j\dots X_n^j$, $Y$, and $W$ are mutually independent, where $\ex(\lambda_j)$, $\ex(\delta)$, and $\ex(\eta)$ denote the exponential distributions with rate $\lambda_j>0$, $\delta>0$, and $\eta>0$, respectively. We have $X\sim \ex(\lambda_j)$ if and only if its  probability density function is given by
\[f(x)=\lambda_j e^{-\lambda_j x }, \quad x>0.\] 
Clearly, $\sum_{i=1}^m X_i^j \sim \gam(m, \lambda_j)$, where $\gam(m, \lambda_j)$ denotes the gamma distribution with shape parameter $m$ and rate parameter $\lambda_j$. That is, $X\sim \gam(m, \lambda_j)$ if and only if its  probability density function and cumulative density function are given by 
\[f(x)=\frac{\lambda_j^m}{\Gamma(m)}x^{m-1}e^{-\lambda_j x}, \quad x>0,\]
and 
\[F(x)=\frac{1}{\Gamma(m)}\gamma(m, \lambda_jx),\]
where $\Gamma(m)=\int_0^\infty t^{m-1}e^{-t} dt$ is the Gamma function.  The cumulative distribution function of $Z_m$ is given by
\begin{eqnarray*}
	F_{Z_m}(z)&=&1-\prod_{j=1}^k \bigg[1-P\bigg(\sum_{i=1}^m X_i^j \leq z\bigg)\bigg ]
	=1-\prod_{j=1}^k\bigg [1-\frac{1}{\Gamma(m)}\gamma(m, \lambda_jz)\bigg ].
\end{eqnarray*}
We have
\begin{align*}
	E[T_m]
=&\frac{\int_0^\infty y\delta e^{-\delta y }\prod_{j=1}^k\bigg [1-\frac{1}{\Gamma(m)}\gamma(m, \lambda_jy)\bigg ]dy}{\int_0^\infty 1-\prod_{j=1}^k\bigg [1-\frac{1}{\Gamma(m)}\gamma(m, \lambda_js)\bigg ]\delta e^{-\delta s }ds
}\\
&+\frac{(1-\int_0^\infty 1-\prod_{j=1}^k\bigg [1-\frac{1}{\Gamma(m)}\gamma(m, \lambda_js)\bigg ]
\delta e^{-\delta y }ds)\int_0^\infty w\eta e^{-\eta w }dw}{\int_0^\infty 1-\prod_{j=1}^k\bigg [1-\frac{1}{\Gamma(m)}\gamma(m, \lambda_js)\bigg ]\delta e^{-\delta s }ds
}\\
&+\frac{\int_0^\infty s(-\prod_{j=1}^k\bigg [1-\frac{1}{\Gamma(m)}\gamma(m, \lambda_js)\bigg ])'e^{-\delta s}ds}{\int_0^\infty 1-\prod_{j=1}^k\bigg [1-\frac{1}{\Gamma(m)}\gamma(m, \lambda_js)\bigg ]\delta e^{-\delta s }ds
},
\end{align*}
and
\begin{align*}
	P_{mk}(t)=&\prod_{j=1}^k\bigg [1-\frac{1}{\Gamma(m)}\gamma(m, \lambda_jt)\bigg ]-(1-e^{-\delta t})\prod_{j=1}^k\bigg [1-\frac{1}{\Gamma(m)}\gamma(m, \lambda_jt)\bigg ]\\
	&+(\frac{P(Y \leq t)}{P(Y\leq Z_m)}1_{\{0\leq t \leq Z_m\}}+1_{\{t > Z_m\}}-\frac{P\{Y+W \leq t\}}{P\{Y\leq Z_m\}}1_{\{0\leq t \leq W+Z_m\}}-1_{\{t >W+Z_m\}})\\
		&\times \int_0^\infty (1-e^{-\delta y})(1-\prod_{j=1}^k\bigg [1-\frac{1}{\Gamma(m)}\gamma(m, \lambda_jy)\bigg ] )'dy\\
&+\int_0^\infty \prod_{j=1}^k\bigg [1-\frac{1}{\Gamma(m)}\gamma(m, \lambda_j(t-s))\bigg ]\\
&-(1-e^{-\delta (t-s)})\prod_{j=1}^k\bigg [1-\frac{1}{\Gamma(m)}\gamma(m, \lambda_j(t-s))\bigg ]dG(s)\\
	&+\int_0^t (\frac{P(Y \leq t-s)}{P(Y\leq Z_m)}1_{\{0\leq t-s \leq Z_m\}}+1_{\{t-s > Z_m\}}-\frac{P\{Y+W \leq t-s\}}{P\{Y\leq Z_m\}}1_{\{0\leq t-s \leq W+Z_m\}}\\
	&-1_{\{t-s >W+Z_m\}})\int_0^\infty (1-e^{-\delta y})(1-\prod_{j=1}^k\bigg [1-\frac{1}{\Gamma(m)}\gamma(m, \lambda_jy)\bigg ] )'dydG(s).
\end{align*}
We perform a simulation study to estimate $E[T_m]$ and $P_{m5}(t)$ as function of $m$ and $t$ for $k=5$.  For  $m=1,2,\dots,40$, we generate $T_{m}$ for the $N=30,000$ iterations. 
Then we estimate $E[T_m]$ as $\frac{1}{N}\sum_{i=1}^N T_{mi}$. Similarly, for $P_{m5}(t)$, we also performed $30, 000$ iterations to estimate each $P_{m5}(t)$ for $m=1,\dots,40$. 
Figure~\ref{fig:exp} displays the values of $E[T_m]$ and $P_{m5}(3)$ under exponential hacking, exponential re-setting, and exponential detecting times for $m=1,2,\dots,40$.  It shows that $E[T_m]$ and $P_{m5}(3)$ are both increasing function of $m$.

\begin{figure}[!ht]
\begin{center}
\subfigure[$E(T_m)$ as a function of $m$]{\centering \includegraphics[scale=0.45]{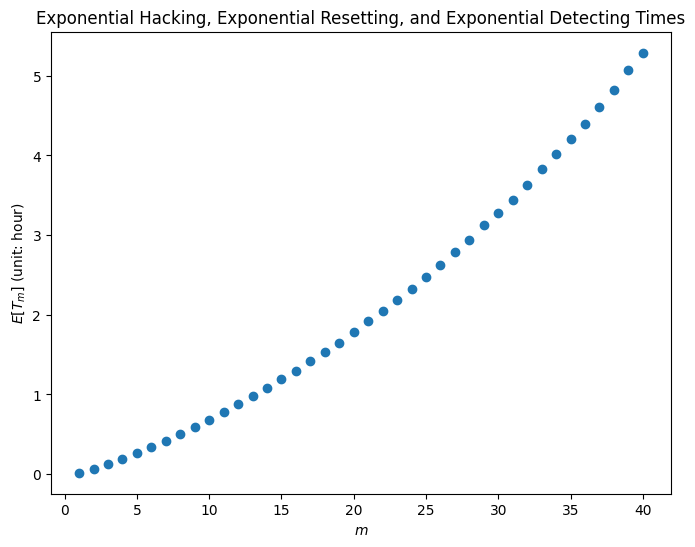}}
\subfigure[$P_{m5}(3)$ as a function of $m$]{\centering \includegraphics[scale=0.45]{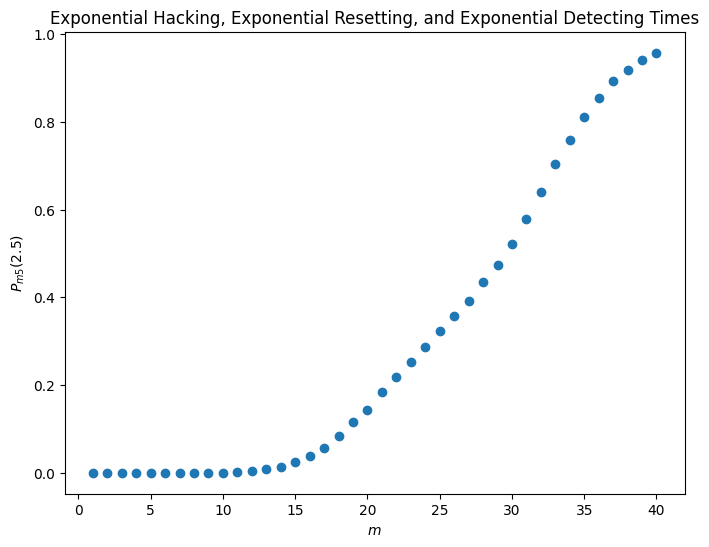}}
\end{center}
\caption{Figure 1: Plots of $E[T_m]$ and $P_{m5}(3)$  under exponential hacking, exponential re-set, and exponential detecting times.}
\label{fig:exp}
\end{figure}

Figure~\ref{fig:exp3D} provides the $3D$ version of $P_{m5}(t)$ and expected net revenue $E_m[NR]$, based on $N=30, 000$ iterations.  Panel $(a)$ of Figure~\ref{fig:exp3D}  plots $P_{m5}(t)$ as a function of $m$ and $t$. 
The blue dotted line in $P_{m5}(t)$ is a $3D$ version of $P_{m5}(3)$.  
Panel $(b)$ of Figure~\ref{fig:exp3D}  demonstrates the expected net revenues per unit of time for various values of $m$ when $R(m)=0.2m, C_1(m)=2m^{0.2},$ and $C_2(m)=2m^{0.3}$.  It is also based on $30,000$ iterations.
The red vertical line shows that the expected net revenue achieves its maximum value of $1.87$ million dollars at $m=40$ for $1\leq m\leq 40$.

\begin{figure}[!ht]
\begin{center}
	 \subfigure[$P_{m5}(t)$ as a bivariate function of $m$ and $t$]{\centering \includegraphics[width=0.53\textwidth]{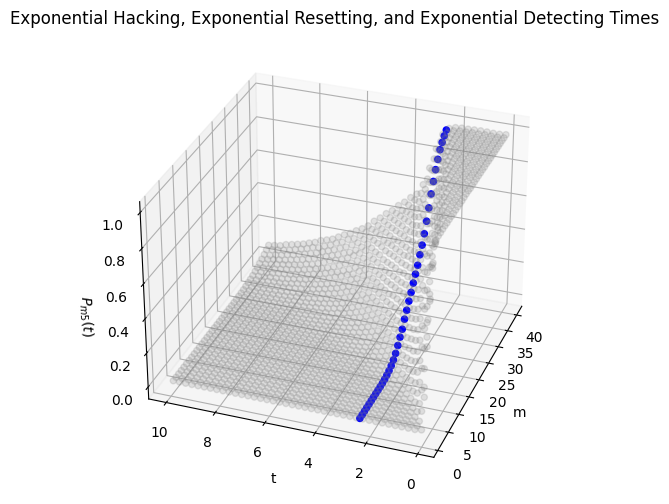}}
	 \subfigure[$E_m(NR)$ as a function of $m$]{\centering \includegraphics[width=0.46\textwidth]{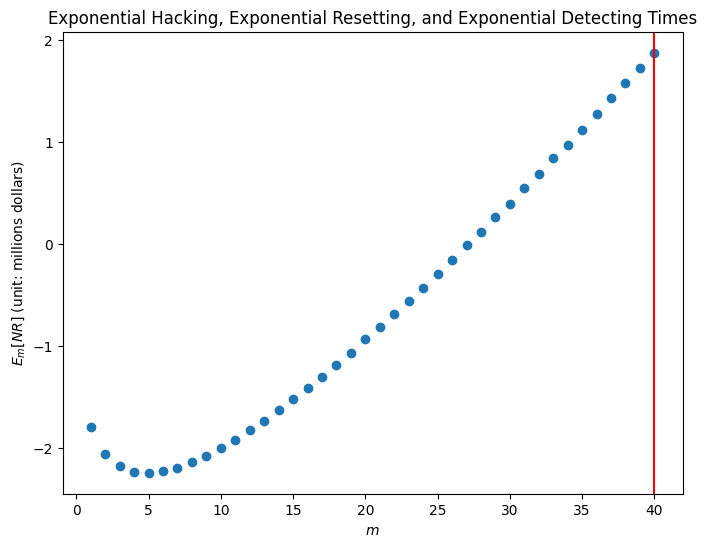}}
\end{center}
\caption{Figure~2: Plots of $P_{m5}(t)$ and  $E_m(NR)$ under exponential hacking, exponential re-set, and exponential detecting times.}
\label{fig:exp3D}
\end{figure}

\end{example}

\newpage
\begin{example}[Gamma hacking, Gamma re-set, and Gamma detecting times]
\label{example2}
Suppose $X_1^j, \ldots,$\\$X_n^j \overset{iid}{\sim} \gam(\eta_j, \delta_j)$,  $Y\sim \gam(\alpha, \beta)$, and $W\sim \gam(\theta, \tau)$ such that $X_1\dots X_n$, $Y$, and $W$ are mutually independent, where $\gam(\eta, \delta)$, $\gam(\alpha, \beta)$, and $\gam(\theta, \tau)$ denote the gamma distributions with shape parameters $\eta>0$, $\alpha>0$, and $\theta>0$ and rate parameters $\delta>0$, $\beta>0$, and $\tau>0$, respectively. Then $\sum_{i=1}^m X_i \sim \gam(m\eta_j, \delta_j)$. 
It follows that
\begin{eqnarray*}
	E[T_m]	&=&\frac{\int_0^\infty y\frac{\beta^\alpha}{\Gamma(\alpha)}y^{\alpha-1}e^{-\beta y} \prod_{j=1}^k \bigg[1-\frac{1}{\Gamma(m\eta_j)}\gamma(m\eta_j, \delta_jz)\bigg ]dy}{\int_0^{\infty}\bigg(1-\prod_{j=1}^k \bigg[1-\frac{1}{\Gamma(m\eta_j)}\gamma(m\eta_j, \delta_js)\bigg ]\bigg)\frac{\beta^\alpha}{\Gamma(\alpha)}s^{\alpha-1}e^{-\beta s}ds}\\
	&&+\frac{\int_0^{\infty}\int_0^{\infty}\bigg(\prod_{j=1}^k \bigg[1-\frac{1}{\Gamma(m\eta_j)}\gamma(m\eta_j, \delta_js)\bigg ]\bigg)\frac{\beta^\alpha}{\Gamma(\alpha)}s^{\alpha-1}e^{-\beta s} w\frac{\tau^\theta}{\Gamma(\theta)}w^{\theta-1}e^{-\tau w}dsdw}{\int_0^{\infty}\bigg(1-\prod_{j=1}^k \bigg[1-\frac{1}{\Gamma(m\eta_j)}\gamma(m\eta_j, \delta_js)\bigg ]\bigg)\frac{\beta^\alpha}{\Gamma(\alpha)}s^{\alpha-1}e^{-\beta s}ds}\\
	&&+\frac{\int_0^\infty (-\prod_{j=1}^k \bigg[1-\frac{1}{\Gamma(m\eta_j)}\gamma(m\eta_j, \delta_js)\bigg ])'(s-\frac{s}{\Gamma(\alpha)}\gamma(\alpha, \beta s))ds}{\int_0^{\infty}\bigg(1-\prod_{j=1}^k \bigg[1-\frac{1}{\Gamma(m\eta_j)}\gamma(m\eta_j, \delta_js)\bigg ]\bigg)\frac{\beta^\alpha}{\Gamma(\alpha)}s^{\alpha-1}e^{-\beta s}ds},
\end{eqnarray*}
and
\begin{eqnarray*}
	P_{mk}(t)&=&\prod_{j=1}^k \bigg[1-\frac{1}{\Gamma(m\eta_j)}\gamma(m\eta_j, \delta_jt)\bigg ]-\frac{1}{\Gamma(\alpha)}\gamma(\alpha, \beta t)\\
	&&+\bigg(1-\prod_{j=1}^k \bigg[1-\frac{1}{\Gamma(m\eta_j)}\gamma(m\eta_j, \delta_jt)\bigg ]\bigg)\frac{1}{\Gamma(\alpha)}\gamma(\alpha, \beta t)\\
	&&+(\frac{\frac{1}{\Gamma(\alpha)}\gamma(\alpha, \beta t)}{P(Y\leq Z_m)}1_{\{0\leq t \leq Z_m\}}+1_{\{t > Z_m\}}-\frac{P\{Y+W \leq t\}}{P\{Y\leq Z_m\}}1_{\{0\leq t \leq W+Z_m\}}-1_{\{t >W+Z_m\}})\\
	&&\times \int_0^{\infty}\frac{1}{\Gamma(\alpha)}\gamma(\alpha, \beta y)(-\prod_{j=1}^k \bigg[1-\frac{1}{\Gamma(m\eta_j)}\gamma(m\eta_j, \delta_jz)\bigg ])'dy\\
	&&+\int_0^t \prod_{j=1}^k \bigg[1-\frac{1}{\Gamma(m\eta_j)}\gamma(m\eta_j, \delta_j(t-s))\bigg ]-\frac{1}{\Gamma(\alpha)}\gamma(\alpha, \beta (t-s))\\
				\end{eqnarray*}
	\begin{eqnarray*}
	&&+(1-\prod_{j=1}^k \bigg[1-\frac{1}{\Gamma(m\eta_j)}\gamma(m\eta_j, \delta_j(t-s))\bigg ])\frac{1}{\Gamma(\alpha)}\gamma(\alpha, \beta (t-s))dG(s)\\
	&&+\int_0^t(\frac{\frac{1}{\Gamma(\alpha)}\gamma(\alpha, \beta (t-s))}{P(Y\leq Z_m)}1_{\{0\leq t-s \leq Z_m\}}+1_{\{t-s > Z_m\}}-\frac{P\{Y+W \leq t-s\}}{P\{Y\leq Z_m\}}1_{\{0\leq t-s \leq W+Z_m\}}\\
	&&-1_{\{t-s >W+Z_m\}})\int_0^{\infty} \frac{1}{\Gamma(\alpha)}\gamma(\alpha, \beta y)(-\prod_{j=1}^k \bigg[1-\frac{1}{\Gamma(m\eta_j)}\gamma(m\eta_j, \delta_jz)\bigg ])'dydG(s).
\end{eqnarray*}
As in Example~1, we performed $N=30,000$ iterations to estimate $E[T_m]$ for $m=1,2,\dots,40$.  
Figure~\ref{fig:gamma} displays plots of $E[T_m]$ and $P_{m4}(2)$ for various values of $m$. Clearly, $P_{m4}(2)$  and $E[T_m]$ are increasing functions of $m$. As $m$ increases, $P_{m4}(2)$ approaches $1$.

\begin{figure}[!ht]
\begin{center}
\subfigure[$E(T_m)$ as a function of $m$]{\centering \includegraphics[scale=0.45]{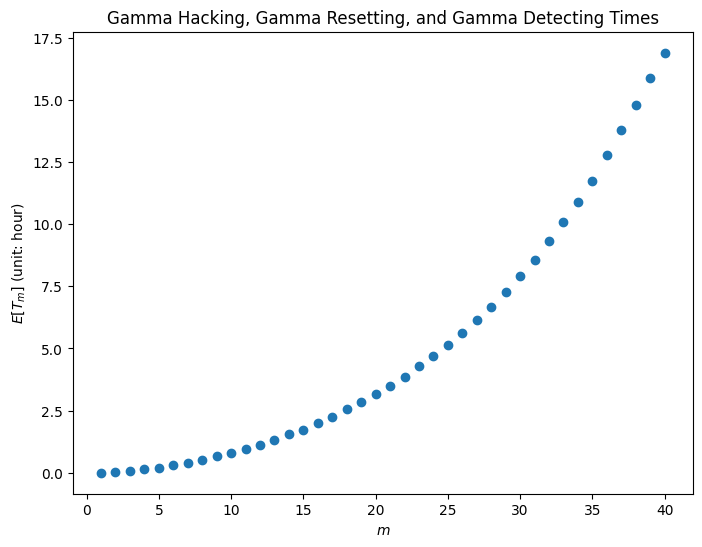}}
\subfigure[$P_{m4}(2)$ as a function of $m$]{\centering \includegraphics[scale=0.45]{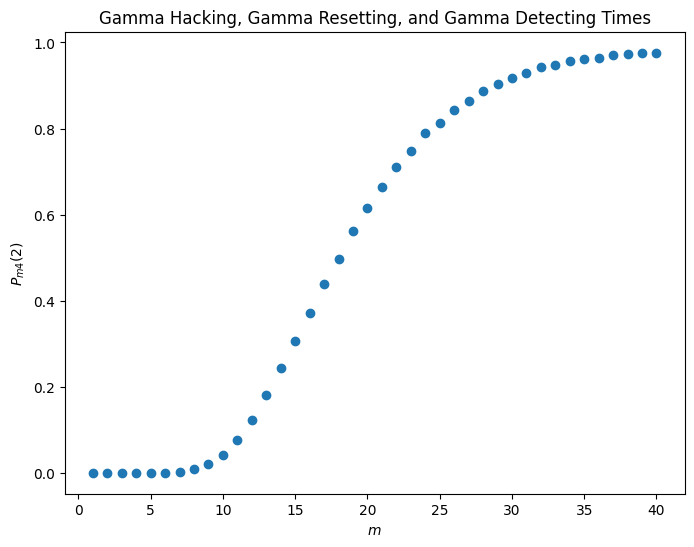}}
\end{center}
\caption{Figure 3: Plots of $E[T_m]$ and $P_{m4}(2)$  under gamma hacking, gamma re-set, and gamma detecting times.}
\label{fig:gamma}
\end{figure} 

\begin{figure}[!ht]
\begin{center}
	 \subfigure[$P_{m4}(t)$ as a bivariate function of $m$ and $t$]{\centering \includegraphics[width=0.53\textwidth]{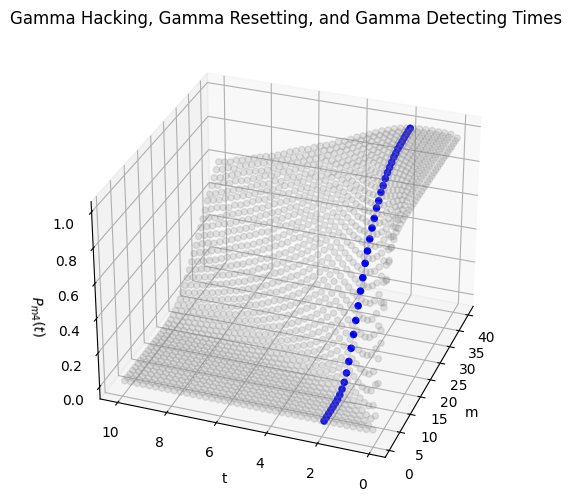}}
	 \subfigure[$E_m(NR)$ as a function of $m$]{\centering \includegraphics[width=0.46\textwidth]{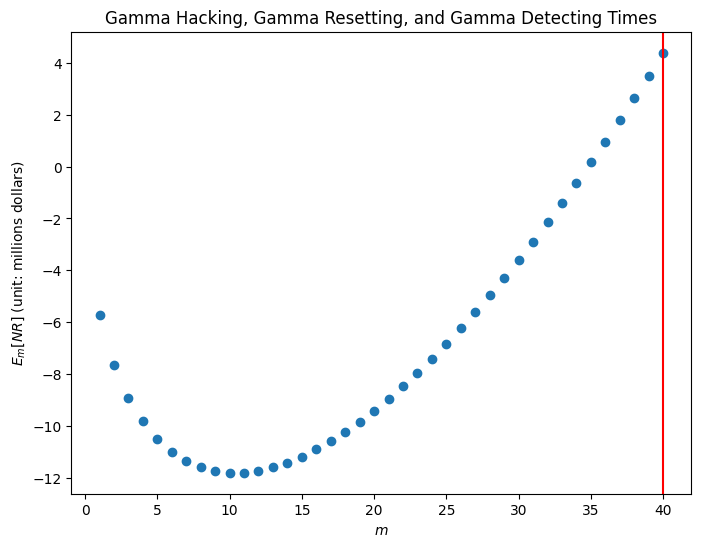}}
\end{center}
\caption{Figure 4: Plots of $P_{m4}(t)$ and  $E_m(NR)$  under gamma hacking, gamma re-set, and gamma detecting times.}
\label{fig:gamma3D}
\end{figure}

Panel (a) of Figure~\ref{fig:gamma3D} provides a 3D plot of $P_{m4}(t)$ as a function of $m$ and $t$ for $k=4$ based on $30, 000$ iteration. The blue dotted line is the 3D version of $P_{m4}(2)$. 
Panel (b) of Figure~\ref{fig:gamma3D} demonstrates $E_m[NR]$ for $R(m)=0.3m+4, C_1(m)=2.5m^2,$ and $C_2(m)=2m^{0.3}$, based on $50,000$ iterations. It shows that $E_m[NR]$ achieves the maximum value of $2.08$ million dollars when $m=1$.
\end{example}

\begin{example}[Gamma hacking, Gamma re-set, and Weibull detecting times]
\label{example3}
The setup here is the same as the one in Example 2 except that here we have $Y\sim\mathsf{Weibull}(\alpha, \beta)$. That is,  the probability density function of $Y$ is given by
\begin{align*}
	f(y)=\frac{\beta}{\alpha}\left(\frac{y}{\alpha}\right)^{\beta-1}e^{-\left(y/\alpha\right)^\beta},\quad y>0.
\end{align*}
In this case, we have
\begin{eqnarray*}
	E[T_m]	&=&\frac{\int_0^\infty y\frac{\beta}{\alpha}\left(\frac{y}{\alpha}\right)^{\beta-1}e^{-\left(y/\alpha\right)^\beta} \prod_{j=1}^k \bigg[1-\frac{1}{\Gamma(m\eta_j)}\gamma(m\eta_j, \delta_jy)\bigg ]dy}{\int_0^{\infty}(1-\prod_{j=1}^k \bigg[1-\frac{1}{\Gamma(m\eta_j)}\gamma(m\eta_j, \delta_js)\bigg ])\frac{\beta}{\alpha}\left(\frac{y}{\alpha}\right)^{\beta-1}e^{-\left(s/\alpha\right)^\beta}ds}\\
	&&+\frac{(1-\int_0^{\infty}(1-\prod_{j=1}^k \bigg[1-\frac{1}{\Gamma(m\eta_j)}\gamma(m\eta_j, \delta_js)\bigg ])f_Y(s)ds)\int_0^\infty w\frac{\tau^\theta}{\Gamma(\theta)}w^{\theta-1}e^{-\tau w}dw}{\int_0^{\infty}(1-\prod_{j=1}^k \bigg[1-\frac{1}{\Gamma(m\eta_j)}\gamma(m\eta_j, \delta_js)\bigg ])\frac{\beta}{\alpha}\left(\frac{y}{\alpha}\right)^{\beta-1}e^{-\left(s/\alpha\right)^\beta}ds}\\
	&&+\frac{\int_0^\infty (-\prod_{j=1}^k \bigg[1-\frac{1}{\Gamma(m\eta_j)}\gamma(m\eta_j, \delta_js)\bigg ])'e^{-(s/\alpha)^\beta}ds}{\int_0^{\infty}(1-\prod_{j=1}^k \bigg[1-\frac{1}{\Gamma(m\eta_j)}\gamma(m\eta_j, \delta_js)\bigg ])\frac{\beta}{\alpha}\left(\frac{y}{\alpha}\right)^{\beta-1}e^{-\left(s/\alpha\right)^\beta}ds}\\
\end{eqnarray*}
and
\begin{eqnarray*}
		P_{mk}(t)	&=&\prod_{j=1}^k \bigg[1-\frac{1}{\Gamma(m\eta_j)}\gamma(m\eta_j, \delta_jt)\bigg ]-1+e^{-(t/\alpha)^\beta}\\
		&&+(1-\prod_{j=1}^k \bigg[1-\frac{1}{\Gamma(m\eta_j)}\gamma(m\eta_j, \delta_jt)\bigg ])(1-e^{-(t/\alpha)^\beta})\\
	&&+(\frac{1-e^{-(t/\alpha)^\beta}}{P(Y\leq Z_m)}1_{\{0\leq t \leq Z_m\}}+1_{\{t > Z_m\}}-\frac{P\{Y+W \leq t\}}{P\{Y\leq Z_m\}}1_{\{0\leq t \leq W+Z_m\}}-1_{\{t >W+Z_m\}})\\
	&&\times\int_0^{\infty}(1-e^{-(y/\alpha)^\beta})(-\prod_{j=1}^k \bigg[1-\frac{1}{\Gamma(m\eta_j)}\gamma(m\eta_j, \delta_jz)\bigg ])'dy\\
	&&+\int_0^t \prod_{j=1}^k \bigg[1-\frac{1}{\Gamma(m\eta_j)}\gamma(m\eta_j, \delta_j(t-s))\bigg ]-1+e^{-((t-s)/\alpha)^\beta}\\
	\end{eqnarray*}
	\begin{eqnarray*}
	&&+(1-\prod_{j=1}^k \bigg[1-\frac{1}{\Gamma(m\eta_j)}\gamma(m\eta_j, \delta_j(t-s))\bigg ])(1-e^{-((t-s)/\alpha)^\beta})dG(s)\\
	&&+\int_0^t(\frac{1-e^{-((t-s)/\alpha)^\beta}}{P(Y\leq Z_m)}1_{\{0\leq t-s \leq Z_m\}}+1_{\{t-s > Z_m\}}\\
	&&-\frac{P\{Y+W \leq t-s\}}{P\{Y\leq Z_m\}}1_{\{0\leq t-s \leq W+Z_m\}}-1_{\{t-s >W+Z_m\}})\\
	&&\times \int_0^{\infty}(1-e^{-(y/\alpha)^\beta})(-\prod_{j=1}^k \bigg[1-\frac{1}{\Gamma(m\eta_j)}\gamma(m\eta_j, \delta_jz)\bigg ])'dydG(s).\\
\end{eqnarray*}
Here we performed $N=20,000$ iterations to estimate $E[T_m]$ for $m=1,2,\dots,40$. 
Figure~\ref{fig:weibull} displays plots of $E[T_m]$ and $P_{m7}(7)$ as $m$ increases.  It shows that $P_{m7}(7)$ and $E[T_m]$ are increasing functions of $m$.

\begin{figure}[!th]
\begin{center}
\subfigure[$E(T_m)$ as a function of $m$]{\centering \includegraphics[scale=0.45]{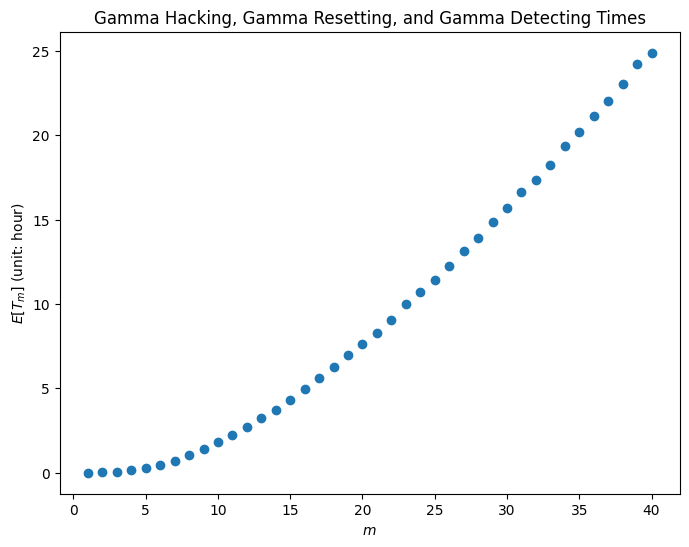}}
\subfigure[$P_{m7}(7)$ as a function of $m$]{\centering \includegraphics[scale=0.45]{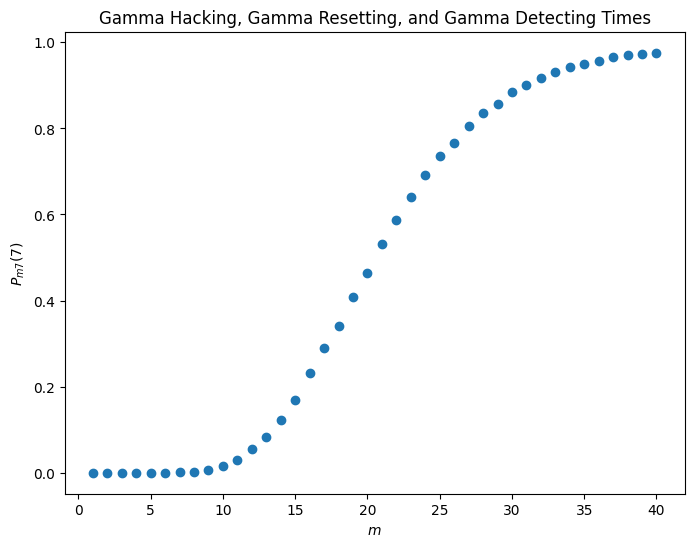}}
\end{center}
\caption{Figure 5: Plots of $E[T_m]$ and $P_{m7}(7)$  under gamma hacking, gamma re-set, and Weibull detecting times.}
\label{fig:weibull}
\end{figure}
\begin{figure}[!ht]
\begin{center}
	 \subfigure[$P_{m7}(t)$ as a bivariate function of $m$ and $t$]{\centering \includegraphics[width=0.53\textwidth]{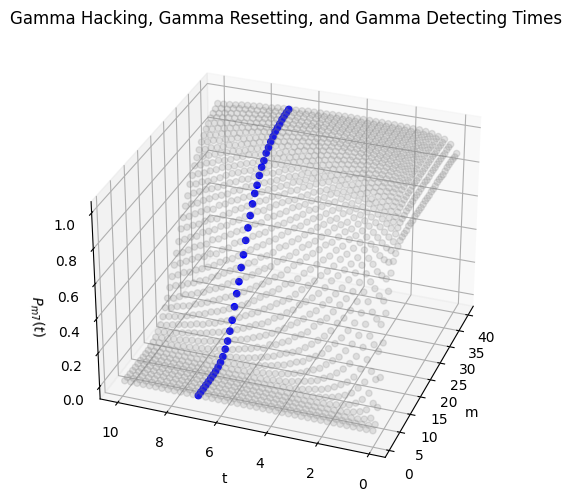}}
	 \subfigure[$E_m(NR)$ as a function of $m$]{\centering \includegraphics[width=0.46\textwidth]{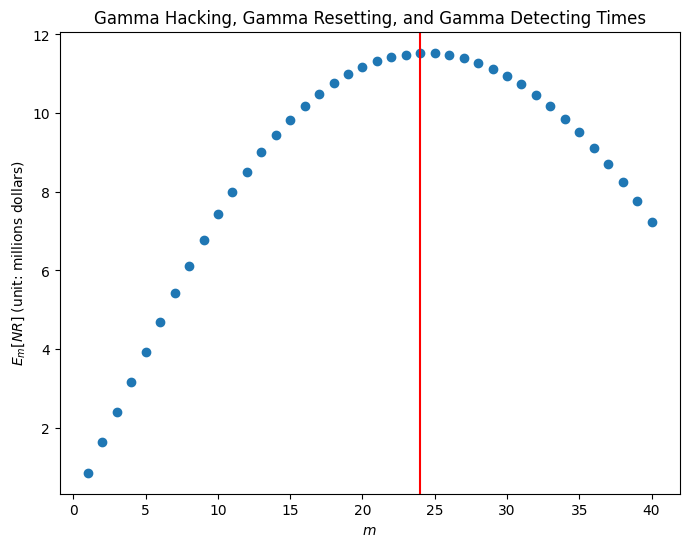}}
\end{center}
\caption{Figure 6: Plots of $P_{m7}(t)$ and $E_m(NR)$  under gamma hacking, gamma re-set, and Weibull detecting times.}
\label{fig:weibull3D}
\end{figure}

Figure~\ref{fig:weibull3D} demonstrates the $3D$ version of $P_{m7}(t)$ and $E_m[NR]$. Panel $(a)$ of Figure~\ref{fig:weibull3D}  provides a plot of $P_{m7}(t)$ as a function of $m$ and $t$, based on $N=20,000$ iterations. 
The blue dotted line is the $3D$ version of $P_{m7}(4)$.  Panel $(b)$ of Figure~\ref{fig:weibull3D}  demonstrates $E_m[NR]$ for different values of $m$ when $R(m)=1.3m^{1.2}, C_1(m)=0.5m^2,$ and $C_2(m)=0.4m^{1.5}$, based on $N=20,000$ iterations. The red vertical line shows that the expected net revenue achieved its maximum value $7.47$ millions dollars at $m=24$.
\end{example}

\newpage
\section{Concluding remarks}

This article studied the stochastic behavior of an $n$-node blockchain under cyber attacks from multiple hackers when re-setting the blockchain takes a random amount of time.  Under the assumption that the detecting times, hacking times, and re-setting times follow arbitrary distributions,  we have derived the instantaneous functional probability,  limiting functional probability,  and mean functional time.  Moreover, we have shown that these three quantities are all increasing functions of the number of nodes,  providing theoretical backup for the intuition that the more nodes a blockchain has, the safer it is.

\section*{Reference}
\begin{description}
\item{} Choi, T.,~ Guo, S.,~Liu, N.,~and Shi, X.~(2020).  Optimal pricing in on-demand-service-platform-operations with hired agents and risk-sensitive customers in the blockchain era. \emph{European Journal of Operational Research}~284(3), 1031--1042.
\item{} Casino, F., Dasaklis, T.~K.~and Patsakis, C.~(2019). A systematic literature review of blockchain-based applications: current status, classification and open issues. \emph{Telematics and Informatics}~36, 55--81.  
\item{} Corbet, S., Cumming, D.~, Lucey, B.~, Peat, M.~, Vigne, S.~(2020). The destabilising effects of cryptocurrency cybercriminality. \emph{Economics Letters}~volume 191.  
\item{} Diffie, W.~and Hellman, M.E.~(1976).  New directions in cryptography. \emph{IEEE Transactions on Information Theory}~22(5), 644--654. 
\item{} Goldwasser, S.~and Micali, S.~(1982). Probabilistic encryption. \emph{Journal of Computer System and Sciences}~28(2), 270--299. 
\item{} Goldwasser, S., Micali, S.~and Rivest, R.L.~(1988). A digital signature scheme secure against adaptive chosen-message attacks. \emph{SIAM Journal on Computing}~18(1), 283--308. 
\item{} Hussain, S., Sivakumar, T.B., and Khang, A.~(2022). Cryptocurrency Methodologies and Techniques. \emph{The Data-Driven Blockchain Ecosystem}~page 9. 
\item{} Khan, M.A.~and Salah, K.~(2017). IoT security: review, blockchain solutions, and open challenges.  \emph{Future Generation Computer Systems.}~82, 395--411.
\item{} Khanum, S.~and Mustafa, K.~(2022). A systematic literature review on sensitive data protection in blockchain applications. \emph{Concurrency and Computation Practice and Experience.}~35(1). 
\item{} Kumar, A.,~ Abhishek, K.,~  Ghalib, M.R.,~  Nerurkar, P.,~ Bhirud, S.,~ Alnumay, W.,~ Kumar, S.A.,~ Chatterjee, P.~and Ghosh, U.~(2020). Securing logistics system and supply chain using Blockchain. \emph{Applied stochastic models in business and industry}~37(3), 413--428. 
\item{} Li, C.~and Zhang, L.J.~(2017). A blockchain based new secure multi-layer network model for internet of things.  \emph{Proceedings--2017 IEEE 2nd Intertional Congress on Internet of Things}, 33--41. 
\item{} Melo, C., Dantas, J.~, Pereira, P.~and Maciel, P.~(2021). Distributed application provisioning over Ethereum-based private and permissioned blockchain: availability modeling, capacity, and costs planning.  \emph{The Journal of Supercomputing}~77, 9615--9641.  
\item{} Meng, W., Tischhauser, E.~, Wang, Q.~, Wang, Y.~and Han, J.~(2017). When instruction detection meets blockchain technology: a review. \emph{IEEE Access}~6, 10179--10188.
\item{} Xu, X.,~and Hong, L.~(2024). Instantaneous and limiting behavior of an n-node blockchain under cyber attacks from a single hacker. Business Analytics in Practice: Enhancing Decision Making (Edited by Emrouznejad et al.) \emph{Lecture Notes in Operations Research}, Springer, New York, to appear.
\item{} Rivest, R.L.~, Shamir, A.~and Adelman, L.~(1978). A method for obtaining digital signatures and public-key cryptosystems.  \emph{Communications of the ACM}~21(2), 120--126.
\item{} Tama, B.A., Kweka, B.J.,~ Park, Y.~and Rhee, K.H.~(2017). A critical review of blockchain and its current applications.  \emph{2017 International Conference on Electrical Engineering and Computer Science IEEE}, 109--113.
\item{} Tsuchiya, Y.~and Hiramoto, N.~(2021). How cryptocurrency is laundered: Case study of Coincheck hacking incident.  \emph{Forensic Science International: Reports}~volume 4.
\item{} Zhang, Z.,~Zargham.,~and M., V.~(2020). On modeling blockchain-enabled economic networks as stochastic dynamical systems.  \emph{Applied Network Science}~5, 19.
\item{} Zheng, Z.,~Xie, S.~, Dai, H.N., and Wang, H.~(2018).  Blockchain challenges and opportunities: a survey.   \emph{International Journal of Web and Grid Services}~14(4), 352--375.
\end{description}

\section*{Appendix}

\subsection*{Some general formulas}

Since
$Z_m=\min_{j=1,\dots,k}\bigg(\sum_{i=1}^m X_i^j\bigg)$,	
we have
\begin{eqnarray*}
	F_{Z_{m}}(z) &=&P\bigg(\min_{j=1,\dots,k}\bigg(\sum_{i=1}^m X_i^j\bigg)\leq z\bigg)\\
	&=&1-P\bigg(\min_{j=1,\dots,k}\bigg(\sum_{i=1}^m X_i^j\bigg)> z\bigg)\\
	&=&1-P\bigg(\sum_{i=1}^m X_i^1 > z, \sum_{i=1}^m X_i^2 > z,\dots, \sum_{i=1}^m X_i^k > z\bigg)\\
	&\stackrel{ind}{=}&1-\prod_{j=1}^k P\bigg(\sum_{i=1}^m X_i^j >z\bigg)\\
	&=&1-\prod_{j=1}^k \bigg[1-P\bigg(\sum_{i=1}^m X_i^j \leq z\bigg)\bigg ].
\end{eqnarray*}
Equation~(5) in Hong and Sarkar (2013) implies
\begin{eqnarray*}
E[Y_1]&=& E\bigg[Y\bigg|Y<Z_m\bigg]=\int_0^\infty yf_{Y|Y<Z_m}(y)dy\\
&=&\frac{\int_0^\infty \int_y^\infty yg(y,s)dsdy}{P(Y<Z_m)}
=\frac{\int_0^\infty \int_y^\infty yf_Y(y)f_{Z_m}(s)dsdy}{1-\int_0^{\infty}F_{Z_m}(s)f_Y(s)ds},
\end{eqnarray*}
where the joint probabiliyt density function of $Y$ and $Z_m$, by independence of $Y$ and $Z_m$, is given by
\begin{align*}
g(y,s)=f_Y(y)f_{Z_m}(s),\quad y>0, s>0. 
\end{align*}
It is clear that
(\ref{eq:detectprob}) implies that
\begin{equation*}
\label{eq:probability}
 p_{mk}=P\bigg\{Y>Z_m\bigg\}=\int_0^{\infty}F_{Z_m}(s)f_Y(s)ds.   
\end{equation*} 
Hence, 
\begin{eqnarray*}
	E[N_1][Y_1]&=&\frac{1-p_{mk}}{p_{mk}}E\bigg[Y\bigg|Y<Z_m\bigg]\\
	&=&\frac{\int_0^\infty \int_y^\infty yf_Y(y)f_{Z_m}(s)dsdy}{\int_0^{\infty}F_{Z_m}(s)f_Y(s)ds}\\
	&=&\frac{\int_0^\infty yf_Y(y) \int_y^\infty f_{Z_m}(s)dsdy}{\int_0^{\infty}F_{Z_m}(s)f_Y(s)ds}\\
	&=&\frac{\int_0^\infty yf_Y(y) (1-F_{Z_m}(y) )dy}{\int_0^{\infty}F_{Z_m}(s)f_Y(s)ds}.
\end{eqnarray*}
Similarly,
\[E[N_1][W_1]=\frac{1-p_{mk}}{p_{mk}}E[W]=\frac{ (1-\int_0^{\infty}F_{Z_m}(s)f_Y(s)ds)\int_0^\infty wf_W(w)dw}{\int_0^{\infty}F_{Z_m}(s)f_Y(s)ds},
\]
and
\begin{align*}
	E\bigg[Z_m\bigg|Z_m<Y\bigg]&=\frac{\int_0^\infty \int_s^\infty sf_Y(y)f_{Z_m}(s)dyds}{\int_0^{\infty}F_{Z_m}(s)f_Y(s)ds}\\
	&=\frac{\int_0^\infty sf_{Z_m}(s) \int_s^\infty f_Y(y)dyds}{\int_0^{\infty}F_{Z_m}(s)f_Y(s)ds}\\
	&=\frac{\int_0^\infty sf_{Z_m}(s) (1-F_Y(s))ds}{\int_0^{\infty}F_{Z_m}(s)f_Y(s)ds}.
\end{align*}
Therefore,
\begin{equation} \label{eq:Tm}
\begin{split}
E[T_m]=&\frac{\int_0^\infty yf_Y(y) (1-F_{Z_m}(y) )dy+(1-\int_0^{\infty}F_{Z_m}(s)f_Y(s)ds)\int_0^\infty wf_W(w)dw}{\int_0^{\infty}F_{Z_m}(s)f_Y(s)ds}\\
	&+\frac{\int_0^\infty sf_{Z_m}(s) (1-F_Y(s))ds}{\int_0^{\infty}F_{Z_m}(s)f_Y(s)ds}.
\end{split}
\end{equation}
Next, we shift our attention to the terms in (\ref{eq:instprob}).  We have
\begin{align*}
P\bigg\{t<Y \wedge Z_m\bigg\}&=1-P\bigg\{Z_m\leq t\bigg\}-P\{Y\leq t\}+P\bigg\{Z_m \leq t,Y\leq t\bigg\}.
\end{align*}
Similarly,
\begin{align*}
	P\bigg\{Y \wedge Z_m>t-s\bigg\}=&1-P\bigg\{Z_m\leq t-s\bigg\}-P\{Y\leq t-s\}\\
	&+P\bigg\{Z_m \leq t-s,Y\leq t-s\bigg\}.
\end{align*}
Also, 
\begin{eqnarray*}
	&&P\left\{Y\leq t<Y+W_1 ,Y\leq Z_m \right\}\\&=&P\left\{Y\leq t<Y+W_1\mid Y\leq Z_m \right\}P\left\{Y\leq Z_m\right\}\\
	&=&P\left\{Y_1\leq t<Y_1+W_1\right\}P\left\{Y\leq Z_m\right\}\\
	&=&\left[1-P\{Y_1>t \text{ or }Y_1+W_1\leq t\}\right]P\left\{Y\leq Z_m\right\}\\
	&=&[1-P\{Y_1>t\}-P\{Y_1+W_1\leq t\}+P\{Y_1>t,Y_1+W_1\leq t\}]P\left\{Y\leq Z_m\right\}\\
	&=&[1-P\{Y_1>t\}-P\{Y_1+W_1\leq t\}]P\left\{Y\leq Z_m\right\}\\
	&=&[P\{Y_1\leq t\}-P\{Y_1+W_1\leq t\}]P\left\{Y\leq Z_m\right\},
\end{eqnarray*}
where $Y_1\overset{d}{=}Y\mid Y\leq Z_m$. Similarly,
\begin{eqnarray*}
	&&P\left\{Y\leq t-s<Y+W_1, Y\leq Z_m\right\}\\
	&=&[P\{Y_1\leq t-s\}-P\{Y_1+W_1\leq t-s\}]P\left\{Y\leq Z_m\right\}.
\end{eqnarray*}
Therefore,
\begin{equation} \label{eq:Pmt}
\begin{split}
	P_{mk}(t)=&1-F_{Z_m}(t)-F_Y(t)+F_{Z_m}(t)F_Y(t)\\
	&+(F_{Y_1}(t)-F_{Y_1+W_1}(t))\int_0^{\infty}F_Y(y)f_{Z_m}(y)dy\\
	&+\int_0^t 1-F_{Z_m}(t-s)-F_Y(t-s)+F_{Z_m}(t-s)F_Y(t-s)dG(s)\\
	&+\int_0^t(F_{Y_1}(t-s)-F_{Y_1+W_1}(t-s))\int_0^{\infty}F_Y(y)f_{Z_m}(y)dydG(s),
\end{split}
\end{equation}
where $G(s)=\sum_{n=1}^\infty F_{\sum_{i=1}^{n} (Y_i+W_i)}(s)$. Also, we have
\begin{equation}
\label{eq:conditionalprobability1}
	F_{Y_1}(t)=\frac{P\bigg\{Y \leq t, Y\leq Z_m\bigg\}}{P\bigg\{Y\leq Z_m\bigg\}}=
	\begin{cases}
		1 & \text{if $t >Z_m$},\\
		\frac{P\{Y \leq t\}}{P\{Y\leq Z_m\}} & \text{if $0\leq t \leq Z_m$},\\
		0 & \text{if $t<0$},
	\end{cases}
\end{equation}
\begin{equation}
\label{eq:conditionalprobability2}
\begin{split}
	F_{Y_1+W_1}(t)&=\frac{P\bigg\{Y+W \leq t, Y\leq Z_m\bigg\}}{P\bigg\{Y\leq Z_m\bigg\}}\\
	&=
	\begin{cases}
		1 & \text{if $t >W+Z_m$},\\
		\frac{P\{Y+W \leq t\}}{P\{Y\leq Z_m\}} & \text{if $0\leq t \leq W+Z_m$},\\
		0 & \text{if $t<0$},
	\end{cases}
	\end{split}
\end{equation}
and
\begin{equation}
\label{eq:conditionalprobability3}
\begin{split}
	F_{\sum_{i=1}^{n} (Y_i+W_i)}(s)&=\frac{P\bigg\{\sum_{i=1}^{n} (Y^\star_i+W_i) \leq s, Y^\star_i\leq Z_{mi}\bigg\}}{P\bigg\{Y^\star_i\leq Z_{mi}\bigg\}}\\
	&=
	\begin{cases}
		1 & \text{if $t >\sum_{i=1}^{n}W_i+\sum_{i=1}^{n}Z_{mi}$},\\
		\frac{P\{\sum_{i=1}^{n} (Y^\star_i+W_i) \leq t\}}{P\{Y\leq Z_{m}\}} & \text{if $0\leq t \leq \sum_{i=1}^{n}W_i+\sum_{i=1}^{n}Z_{mi}$},\\
		0 & \text{if $t<0$},
	\end{cases}
	\end{split}
\end{equation}
where $Y^\star_i \overset{iid}{\sim}F_{Y}(y)$ and $Z_{mi}$ denotes the minimum time spend on hacking into the $m$ node during the $ith$ cycle. Thus, (\ref{eq:conditionalprobability1}), (\ref{eq:conditionalprobability2}), and (\ref{eq:conditionalprobability3}) can also be written as
\begin{eqnarray*}
	F_{Y_1}(t)&=&\frac{P(Y \leq t)}{P(Y\leq Z_m)}1_{\{0\leq t \leq Z_m\}}+1_{\{t > Z_m\}},\\
	F_{Y_1+W_1}(t)&=&\frac{P\{Y+W \leq t\}}{P\{Y\leq Z_m\}}1_{\{0\leq t \leq W+Z_m\}}+1_{\{t >W+Z_m\}},\\
F_{\sum_{i=1}^{n} (Y_i+W_i)}(s)&=&\frac{P\{\sum_{i=1}^{n} (Y^\star_i+W_i) \leq t\}}{P\{Y\leq Z_m\}} 1_{\{0\leq t \leq \sum_{i=1}^{n}W_i+\sum_{i=1}^{n}Z_{mi}\}}\\
&&+1_{\{t >\sum_{i=1}^{n}W_i+\sum_{i=1}^{n}Z_{mi}\}}.
\end{eqnarray*}

\subsection*{From Example~\ref{example1}}
Since $Y\sim\ex(\delta)$ and $W\sim\ex(\eta)$,  (\ref{eq:probability}) implies
\begin{align*}
	p_{mk}=P\bigg\{Y>Z_m\bigg\}&=\int_0^{\infty}\bigg(1-\prod_{j=1}^k\bigg [1-\frac{1}{\Gamma(m)}\gamma(m, \lambda_js)\bigg ]\bigg)\delta e^{-\delta s}ds,
\end{align*}
where $\gamma(m, \lambda_j s)=\int_0^{\lambda_j s}t^{m-1}e^{-t}dt$ is the lower incomplete gamma function.
By  (\ref{eq:Tm}), we have
\begin{align*}
	E[T_m]=&\frac{\int_0^\infty yf_Y(y) (1-F_{Z_m}(y) )dy+(1-\int_0^{\infty}F_{Z_m}(s)f_Y(s)ds)\int_0^\infty wf_W(w)dw}{\int_0^{\infty}F_{Z_m}(s)f_Y(s)ds}\\
	&+\frac{\int_0^\infty sf_{Z_m}(s) (1-F_Y(s))ds}{\int_0^{\infty}F_{Z_m}(s)f_Y(s)ds}\\
	=&\frac{\int_0^\infty y\delta e^{-\delta y }(1-(1-\prod_{j=1}^k\bigg [1-\frac{1}{\Gamma(m)}\gamma(m, \lambda_jy)\bigg ]))dy}{\int_0^\infty 1-\prod_{j=1}^k\bigg [1-\frac{1}{\Gamma(m)}\gamma(m, \lambda_js)\bigg ]\delta e^{-\delta s }ds
}
	\end{align*}
	
\begin{align*}
	&+\frac{(1-\int_0^\infty 1-\prod_{j=1}^k\bigg [1-\frac{1}{\Gamma(m)}\gamma(m, \lambda_js)\bigg ]
\delta e^{-\delta y }ds)\int_0^\infty w\eta e^{-\eta w }dw}{\int_0^\infty 1-\prod_{j=1}^k\bigg [1-\frac{1}{\Gamma(m)}\gamma(m, \lambda_js)\bigg ]\delta e^{-\delta s }ds
}\\
&+\frac{\int_0^\infty s(-\prod_{j=1}^k\bigg [1-\frac{1}{\Gamma(m)}\gamma(m, \lambda_js)\bigg ])'(1-(1-e^{-\delta s}))ds}{\int_0^\infty 1-\prod_{j=1}^k\bigg [1-\frac{1}{\Gamma(m)}\gamma(m, \lambda_js)\bigg ]\delta e^{-\delta s }ds
}\\
=&\frac{\int_0^\infty y\delta e^{-\delta y }\prod_{j=1}^k\bigg [1-\frac{1}{\Gamma(m)}\gamma(m, \lambda_jy)\bigg ]dy}{\int_0^\infty 1-\prod_{j=1}^k\bigg [1-\frac{1}{\Gamma(m)}\gamma(m, \lambda_js)\bigg ]\delta e^{-\delta s }ds
}\\
&+\frac{(1-\int_0^\infty 1-\prod_{j=1}^k\bigg [1-\frac{1}{\Gamma(m)}\gamma(m, \lambda_js)\bigg ]
\delta e^{-\delta y }ds)\int_0^\infty w\eta e^{-\eta w }dw}{\int_0^\infty 1-\prod_{j=1}^k\bigg [1-\frac{1}{\Gamma(m)}\gamma(m, \lambda_js)\bigg ]\delta e^{-\delta s }ds
}\\
&+\frac{\int_0^\infty s(-\prod_{j=1}^k\bigg [1-\frac{1}{\Gamma(m)}\gamma(m, \lambda_js)\bigg ])'e^{-\delta s}ds}{\int_0^\infty 1-\prod_{j=1}^k\bigg [1-\frac{1}{\Gamma(m)}\gamma(m, \lambda_js)\bigg ]\delta e^{-\delta s }ds
}.
\end{align*}
By  (\ref{eq:Pmt}), we have
\begin{eqnarray*}
	P_{mk}(t)
	&=&1-(1-\prod_{j=1}^k\bigg [1-\frac{1}{\Gamma(m)}\gamma(m, \lambda_jt)\bigg ])-(1-e^{-\delta t})\\
	&&+(1-\prod_{j=1}^k\bigg [1-\frac{1}{\Gamma(m)}\gamma(m, \lambda_jt)\bigg ])(1-e^{-\delta t})\\
	&&+(\frac{1-e^{-\delta t}}{P(Y\leq Z_m)}1_{\{0\leq t \leq Z_m\}}+1_{\{t > Z_m\}}-\frac{P\{Y+W \leq t\}}{P\{Y\leq Z_m\}}1_{\{0\leq t \leq W+Z_m\}}-1_{\{t >W+Z_m\}})\\
		&&\times \int_0^\infty (1-e^{-\delta y})(1-\prod_{j=1}^k\bigg [1-\frac{1}{\Gamma(m)}\gamma(m, \lambda_jy)\bigg ] )'dy\\
	&&+\int_0^\infty 1-(1-\prod_{j=1}^k\bigg [1-\frac{1}{\Gamma(m)}\gamma(m, \lambda_j(t-s))\bigg ])-(1-e^{-\delta (t-s)})\\
	&&+(1-\prod_{j=1}^k\bigg [1-\frac{1}{\Gamma(m)}\gamma(m, \lambda_j(t-s))\bigg ])(1-e^{-\delta (t-s)})dG(s)\\
	\end{eqnarray*}
	\begin{eqnarray*}
	&&+\int_0^t (\frac{1-e^{-\delta (t-s)}}{P(Y\leq Z_m)}1_{\{0\leq t-s \leq Z_m\}}+1_{\{t-s > Z_m\}}-\frac{P\{Y+W \leq t-s\}}{P\{Y\leq Z_m\}}1_{\{0\leq t-s \leq W+Z_m\}}\\
	&&-1_{\{t-s >W+Z_m\}})\int_0^\infty (1-e^{-\delta y})(1-\prod_{j=1}^k\bigg [1-\frac{1}{\Gamma(m)}\gamma(m, \lambda_jy)\bigg ] )'dydG(s)\\
	&=&\prod_{j=1}^k\bigg [1-\frac{1}{\Gamma(m)}\gamma(m, \lambda_jt)\bigg ]-(1-e^{-\delta t})\prod_{j=1}^k\bigg [1-\frac{1}{\Gamma(m)}\gamma(m, \lambda_jt)\bigg ]\\
	&&+(\frac{1-e^{-\delta t}}{P(Y\leq Z_m)}1_{\{0\leq t \leq Z_m\}}+1_{\{t > Z_m\}}-\frac{P\{Y+W \leq t\}}{P\{Y\leq Z_m\}}1_{\{0\leq t \leq W+Z_m\}}-1_{\{t >W+Z_m\}})\\
		&&\times \int_0^\infty (1-e^{-\delta y})(1-\prod_{j=1}^k\bigg [1-\frac{1}{\Gamma(m)}\gamma(m, \lambda_jy)\bigg ] )'dy\\
	&&+\int_0^\infty \prod_{j=1}^k\bigg [1-\frac{1}{\Gamma(m)}\gamma(m, \lambda_j(t-s))\bigg ]-(1-e^{-\delta (t-s)})\prod_{j=1}^k\bigg [1-\frac{1}{\Gamma(m)}\gamma(m, \lambda_j(t-s))\bigg ]dG(s)\\
	&&+\int_0^t (\frac{1-e^{-\delta (t-s)}}{P(Y\leq Z_m)}1_{\{0\leq t-s \leq Z_m\}}+1_{\{t-s > Z_m\}}-\frac{P\{Y+W \leq t-s\}}{P\{Y\leq Z_m\}}1_{\{0\leq t-s \leq W+Z_m\}}\\
	&&-1_{\{t-s >W+Z_m\}})\int_0^\infty (1-e^{-\delta y})(1-\prod_{j=1}^k\bigg [1-\frac{1}{\Gamma(m)}\gamma(m, \lambda_jy)\bigg ] )'dydG(s),
\end{eqnarray*}
where 
\begin{eqnarray*}
	\gamma(m, \lambda_j s)=\int_0^{\lambda_j s}z^{m-1}e^{-z}dz,
\end{eqnarray*}
\begin{eqnarray*}
	G(s)&=&\sum_{n=1}^\infty\frac{P\{\sum_{i=1}^{n} (Y^\star_i+W_i) \leq t\}}{\int_0^\infty (1-e^{-\delta y})(1-\prod_{j=1}^k\bigg [1-\frac{1}{\Gamma(m)}\gamma(m, \lambda_jt)\bigg ] )'dy}1_{0\leq t \leq \sum_{i=1}^{n}W_i+\sum_{i=1}^{n}Z_{mi}}\\
	&&+1_{t >\sum_{i=1}^{n}W_i+\sum_{i=1}^{n}Z_{mi}},
\end{eqnarray*}
and 
\begin{eqnarray*}
	f_{Z_m}(y)=\frac{F_{Z_m}(y)}{dy}=(1-\prod_{j=1}^k\bigg [1-\frac{1}{\Gamma(m)}\gamma(m, \lambda_jy)\bigg ] )'.
\end{eqnarray*}
Note that $P\{\sum_{i=1}^{n} (Y^\star_i+W_i) \leq t\}$ is the distribution function of the sum of two independent gamma random variables, $\sum_{i=1}^n Y_i ^\star \sim \gam(n, \delta)$ and  $\sum_{i=1}^n W_i\sim \gam(n, \eta)$.  Thus, Moschopoulos~(1985) implies
\begin{eqnarray*}
	P\{\sum_{i=1}^{n} (Y^\star_i+W_i) \leq t\}=C\sum_{k=0}^\infty \delta_l \int_0^t \frac{y^{2n+k-1}e^{-y(\delta \vee \eta)}(\delta \vee \eta)^{2n+l}}{\Gamma(2n+l)}dy,
\end{eqnarray*}
where
\begin{eqnarray*}
	C&=&\frac{(\delta \vee \eta)^{2n}}{(\delta\eta)^n},\\
	\delta_{l+1}&=&\frac{1}{l+1}\sum_{i=1}^{l+1}i\gamma_i\delta_{l+1-i},\\
\end{eqnarray*}
for $l=0,1,2,\dots$, and
\begin{eqnarray*}
	\gamma_j=\frac{n}{j}((1-\frac{\delta}{\delta \vee \eta})^j+(1-\frac{\eta}{\delta \vee \eta})^j),
\end{eqnarray*}
for $j=1,2,\dots$.
Moreover, 
\begin{eqnarray*}
	P(Y+W \leq t)=1-\frac{\eta}{\eta-\delta}e^{-\delta t}+\frac{\delta}{\eta-\delta}e^{-\eta t},
\end{eqnarray*}
and 
\begin{eqnarray*}
	P(Y\leq Z_m)&=&1-\int_0^{\infty}\bigg(1-\prod_{j=1}^k\bigg [1-\frac{1}{\Gamma(m)}\gamma(m, \lambda_js)\bigg ]\bigg)\delta e^{-\delta s}ds.\\
\end{eqnarray*}

\subsection*{From Example~\ref{example2}}

$\sum_{i=1}^m X_i^j \sim \gam(m\eta_j, \delta_j)$, $Y\sim \gam(\alpha, \beta)$, and $W\sim \gam(\theta, \tau)$. It is clear that 
\begin{eqnarray*}
	F_{Z_{m}}(z) 
	&=&1-\prod_{j=1}^k \bigg[1-P\bigg(\sum_{i=1}^m X_i^j \leq z\bigg)\bigg ]\\
	&=&1-\prod_{j=1}^k \bigg[1-\frac{1}{\Gamma(m\eta_j)}\gamma(m\eta_j, \delta_jz)\bigg ],
\end{eqnarray*}
\begin{eqnarray*}
	f_Y(y)=\frac{\beta^\alpha}{\Gamma(\alpha)}y^{\alpha-1}e^{-\beta y}, 
\end{eqnarray*}
\begin{eqnarray*}
	F_Y(y) = \frac{1}{\Gamma(\alpha)}\gamma(\alpha, \beta y),
\end{eqnarray*}
\begin{eqnarray*}
	f_W(w)=\frac{\tau^\theta}{\Gamma(\theta)}w^{\theta-1}e^{-\tau w},
\end{eqnarray*}
and
\begin{eqnarray*}
	F_W(w) = \frac{1}{\Gamma(\theta)}\gamma(\theta, \tau w).
\end{eqnarray*}
By  (\ref{eq:Tm}), we have
\begin{eqnarray*}
	E[T_m]&=&\frac{\int_0^\infty y\frac{\beta^\alpha}{\Gamma(\alpha)}y^{\alpha-1}e^{-\beta y} (1-(1-\prod_{j=1}^k \bigg[1-\frac{1}{\Gamma(m\eta_j)}\gamma(m\eta_j, \delta_jz)\bigg ]) )dy}{\int_0^{\infty}\bigg(1-\prod_{j=1}^k \bigg[1-\frac{1}{\Gamma(m\eta_j)}\gamma(m\eta_j, \delta_js)\bigg ]\bigg)\frac{\beta^\alpha}{\Gamma(\alpha)}s^{\alpha-1}e^{-\beta s}ds}\\
	&&+\frac{(1-\int_0^{\infty}\bigg(1-\prod_{j=1}^k \bigg[1-\frac{1}{\Gamma(m\eta_j)}\gamma(m\eta_j, \delta_js)\bigg ]\bigg)\frac{\beta^\alpha}{\Gamma(\alpha)}s^{\alpha-1}e^{-\beta s}ds)\int_0^\infty w\frac{\tau^\theta}{\Gamma(\theta)}w^{\theta-1}e^{-\tau w}dw}{\int_0^{\infty}\bigg(1-\prod_{j=1}^k \bigg[1-\frac{1}{\Gamma(m\eta_j)}\gamma(m\eta_j, \delta_js)\bigg ]\bigg)\frac{\beta^\alpha}{\Gamma(\alpha)}s^{\alpha-1}e^{-\beta s}ds}\\
	&&+\frac{\int_0^\infty s(-\prod_{j=1}^k \bigg[1-\frac{1}{\Gamma(m\eta_j)}\gamma(m\eta_j, \delta_js)\bigg ])'(1-\frac{1}{\Gamma(\alpha)}\gamma(\alpha, \beta s))ds}{\int_0^{\infty}\bigg(1-\prod_{j=1}^k \bigg[1-\frac{1}{\Gamma(m\eta_j)}\gamma(m\eta_j, \delta_js)\bigg ]\bigg)\frac{\beta^\alpha}{\Gamma(\alpha)}s^{\alpha-1}e^{-\beta s}ds}\\
	&=&\frac{\int_0^\infty y\frac{\beta^\alpha}{\Gamma(\alpha)}y^{\alpha-1}e^{-\beta y} \prod_{j=1}^k \bigg[1-\frac{1}{\Gamma(m\eta_j)}\gamma(m\eta_j, \delta_jz)\bigg ]dy}{\int_0^{\infty}\bigg(1-\prod_{j=1}^k \bigg[1-\frac{1}{\Gamma(m\eta_j)}\gamma(m\eta_j, \delta_js)\bigg ]\bigg)\frac{\beta^\alpha}{\Gamma(\alpha)}s^{\alpha-1}e^{-\beta s}ds}\\
	&&+\frac{\int_0^{\infty}\int_0^{\infty}\bigg(\prod_{j=1}^k \bigg[1-\frac{1}{\Gamma(m\eta_j)}\gamma(m\eta_j, \delta_js)\bigg ]\bigg)\frac{\beta^\alpha}{\Gamma(\alpha)}s^{\alpha-1}e^{-\beta s} w\frac{\tau^\theta}{\Gamma(\theta)}w^{\theta-1}e^{-\tau w}dsdw}{\int_0^{\infty}\bigg(1-\prod_{j=1}^k \bigg[1-\frac{1}{\Gamma(m\eta_j)}\gamma(m\eta_j, \delta_js)\bigg ]\bigg)\frac{\beta^\alpha}{\Gamma(\alpha)}s^{\alpha-1}e^{-\beta s}ds}\\
	&&+\frac{\int_0^\infty (-\prod_{j=1}^k \bigg[1-\frac{1}{\Gamma(m\eta_j)}\gamma(m\eta_j, \delta_js)\bigg ])'(s-\frac{s}{\Gamma(\alpha)}\gamma(\alpha, \beta s))ds}{\int_0^{\infty}\bigg(1-\prod_{j=1}^k \bigg[1-\frac{1}{\Gamma(m\eta_j)}\gamma(m\eta_j, \delta_js)\bigg ]\bigg)\frac{\beta^\alpha}{\Gamma(\alpha)}s^{\alpha-1}e^{-\beta s}ds}.\\
	P_{mk}(t)&=&1-(1-\prod_{j=1}^k \bigg[1-\frac{1}{\Gamma(m\eta_j)}\gamma(m\eta_j, \delta_jt)\bigg ])-\frac{1}{\Gamma(\alpha)}\gamma(\alpha, \beta t)\\
	&&+\bigg(1-\prod_{j=1}^k \bigg[1-\frac{1}{\Gamma(m\eta_j)}\gamma(m\eta_j, \delta_jt)\bigg ]\bigg)\frac{1}{\Gamma(\alpha)}\gamma(\alpha, \beta t)\\
	&&+(\frac{\frac{1}{\Gamma(\alpha)}\gamma(\alpha, \beta t)}{P(Y\leq Z_m)}1_{\{0\leq t \leq Z_m\}}+1_{\{t > Z_m\}}-\frac{P\{Y+W \leq t\}}{P\{Y\leq Z_m\}}1_{\{0\leq t \leq W+Z_m\}}-1_{\{t >W+Z_m\}})\\
		\end{eqnarray*}
\begin{eqnarray*}
	&&\times \int_0^{\infty}\frac{1}{\Gamma(\alpha)}\gamma(\alpha, \beta y)(-\prod_{j=1}^k \bigg[1-\frac{1}{\Gamma(m\eta_j)}\gamma(m\eta_j, \delta_jz)\bigg ])'dy\\
	&&+\int_0^t 1-(1-\prod_{j=1}^k \bigg[1-\frac{1}{\Gamma(m\eta_j)}\gamma(m\eta_j, \delta_j(t-s))\bigg ])-\frac{1}{\Gamma(\alpha)}\gamma(\alpha, \beta (t-s))\\
	&&+(1-\prod_{j=1}^k \bigg[1-\frac{1}{\Gamma(m\eta_j)}\gamma(m\eta_j, \delta_j(t-s))\bigg ])\frac{1}{\Gamma(\alpha)}\gamma(\alpha, \beta (t-s))dG(s)\\
	&&+\int_0^t(\frac{\frac{1}{\Gamma(\alpha)}\gamma(\alpha, \beta (t-s))}{P(Y\leq Z_m)}1_{\{0\leq t-s \leq Z_m\}}+1_{\{t-s > Z_m\}}-\frac{P\{Y+W \leq t-s\}}{P\{Y\leq Z_m\}}1_{\{0\leq t-s \leq W+Z_m\}}\\
	&&-1_{\{t-s >W+Z_m\}})\int_0^{\infty} \frac{1}{\Gamma(\alpha)}\gamma(\alpha, \beta y)(-\prod_{j=1}^k \bigg[1-\frac{1}{\Gamma(m\eta_j)}\gamma(m\eta_j, \delta_jz)\bigg ])'dydG(s)\\
	&=&\prod_{j=1}^k \bigg[1-\frac{1}{\Gamma(m\eta_j)}\gamma(m\eta_j, \delta_jt)\bigg ]-\frac{1}{\Gamma(\alpha)}\gamma(\alpha, \beta t)\\
	&&+\bigg(1-\prod_{j=1}^k \bigg[1-\frac{1}{\Gamma(m\eta_j)}\gamma(m\eta_j, \delta_jt)\bigg ]\bigg)\frac{1}{\Gamma(\alpha)}\gamma(\alpha, \beta t)\\
	&&+(\frac{\frac{1}{\Gamma(\alpha)}\gamma(\alpha, \beta t)}{P(Y\leq Z_m)}1_{\{0\leq t \leq Z_m\}}+1_{\{t > Z_m\}}-\frac{P\{Y+W \leq t\}}{P\{Y\leq Z_m\}}1_{\{0\leq t \leq W+Z_m\}}-1_{\{t >W+Z_m\}})\\
	&&\times \int_0^{\infty}\frac{1}{\Gamma(\alpha)}\gamma(\alpha, \beta y)(-\prod_{j=1}^k \bigg[1-\frac{1}{\Gamma(m\eta_j)}\gamma(m\eta_j, \delta_jz)\bigg ])'dy\\
	&&+\int_0^t \prod_{j=1}^k \bigg[1-\frac{1}{\Gamma(m\eta_j)}\gamma(m\eta_j, \delta_j(t-s))\bigg ]-\frac{1}{\Gamma(\alpha)}\gamma(\alpha, \beta (t-s))\\
	&&+(1-\prod_{j=1}^k \bigg[1-\frac{1}{\Gamma(m\eta_j)}\gamma(m\eta_j, \delta_j(t-s))\bigg ])\frac{1}{\Gamma(\alpha)}\gamma(\alpha, \beta (t-s))dG(s)\\
	&&+\int_0^t(\frac{\frac{1}{\Gamma(\alpha)}\gamma(\alpha, \beta (t-s))}{P(Y\leq Z_m)}1_{\{0\leq t-s \leq Z_m\}}+1_{\{t-s > Z_m\}}-\frac{P\{Y+W \leq t-s\}}{P\{Y\leq Z_m\}}1_{\{0\leq t-s \leq W+Z_m\}}\\
	&&-1_{\{t-s >W+Z_m\}})\int_0^{\infty} \frac{1}{\Gamma(\alpha)}\gamma(\alpha, \beta y)(-\prod_{j=1}^k \bigg[1-\frac{1}{\Gamma(m\eta_j)}\gamma(m\eta_j, \delta_jz)\bigg ])'dydG(s),
\end{eqnarray*}
where
\begin{eqnarray*}
	P\{Y+W \leq t\}=C\sum_{k=0}^\infty \delta_k \int_0^t \frac{y^{\alpha+\theta+k-1}e^{-y(\beta \vee \tau)}(\beta \vee \tau)^{\alpha+\theta+k}}{\Gamma(\alpha+\theta+k)}dy,
\end{eqnarray*}
and
\begin{eqnarray*}
	C&=&\frac{(\beta \vee \tau)^{2n}}{(\beta \tau)^n},\\
	\delta_{k+1}&=&\frac{1}{k+1}\sum_{i=1}^{k+1}i\gamma_i\delta_{k+1-i},
\end{eqnarray*}
for $k=0,1,2,\dots$, 
\begin{eqnarray*}
	\gamma_j=\frac{\alpha}{j}(1-\frac{\beta}{\beta \vee \tau})^j+\frac{\theta}{j}(1-\frac{\tau}{\beta \vee \tau})^j
\end{eqnarray*}
for $j=1,2,\dots$, and
\begin{eqnarray*}
	G(S)&=&\sum_{n=1}^\infty\frac{P\{\sum_{i=1}^{n} (Y^\star_i+W_i) \leq t\}}{P\{Y\leq Z_m\}} 1_{\{0\leq t \leq \sum_{i=1}^{n}W_i+\sum_{i=1}^{n}Z_{mi}\}}\\
&&+1_{\{t >\sum_{i=1}^{n}W_i+\sum_{i=1}^{n}Z_{mi}\}}.\\
\end{eqnarray*}
Since $\sum_{i=1}^n Y_i^\star\sim \gam(n\alpha, \beta)$ and $\sum_{i=1}^{n}W_i\sim \gam(n\theta, \tau)$,
\begin{eqnarray*}
	P\{\sum_{i=1}^{n} (Y^\star_i+W_i) \leq t\}=C\sum_{k=0}^\infty \delta_k \int_0^t \frac{y^{n\alpha+n\theta+k-1}e^{-y(\beta \vee \tau)}(\beta \vee \tau)^{n\alpha+n\theta+k}}{\Gamma(n\alpha+n\theta+k)}dy,
\end{eqnarray*}
where
\begin{eqnarray*}
	C&=&\frac{(\beta \vee \tau)^{2n}}{(\beta \tau)^n},\\
	\delta_{k+1}&=&\frac{1}{k+1}\sum_{i=1}^{k+1}i\gamma_i\delta_{k+1-i},
\end{eqnarray*}
for $k=0,1,2,\dots$, and
\begin{eqnarray*}
	\gamma_j=\frac{n\alpha}{j}(1-\frac{\beta}{\beta \vee \tau})^j+\frac{n\theta}{j}(1-\frac{\tau}{\beta \vee \tau})^j,
\end{eqnarray*}
for $j=1,2,\dots$.Moreover, 
\begin{eqnarray*}
	P(Y \leq Z_m)&=&1-\int_0^{\infty} F_{Z_m}(s) f_{Y}(s)ds\\
	&=&1-\int_0^\infty(1-\prod_{j=1}^k \bigg[1-\frac{1}{\Gamma(m\eta_j)}\gamma(m\eta_j, \delta_js)\bigg ])\frac{\beta^\alpha}{\Gamma(\alpha)}s^{\alpha-1}e^{-\beta s}ds.
\end{eqnarray*}

\subsection*{From Example~\ref{example3}}
Since $\sum_{i=1}^m X_i^j \sim \gam(m\eta_j, \delta_j)$, $Y\sim\mathsf{Weibull}(\alpha, \beta)$, and $W\sim \gam(\theta, \tau)$. We have
\begin{eqnarray*}
	F_{Z_{m}}(z) 
	&=&1-\prod_{j=1}^k \bigg[1-P\bigg(\sum_{i=1}^m X_i^j \leq z\bigg)\bigg ]\\
	&=&1-\prod_{j=1}^k \bigg[1-\frac{1}{\Gamma(m\eta_j)}\gamma(m\eta_j, \delta_jz)\bigg ],
\end{eqnarray*}
\begin{eqnarray*}
	f_Y(y)=\frac{\beta}{\alpha}\left(\frac{y}{\alpha}\right)^{\beta-1}e^{-\left(y/\alpha\right)^\beta},
\end{eqnarray*}
\begin{eqnarray*}
	F_Y(y) = 1-e^{-(y/\alpha)^\beta},
\end{eqnarray*}
\begin{eqnarray*}
	f_W(w)=\frac{\tau^\theta}{\Gamma(\theta)}w^{\theta-1}e^{-\tau w},
\end{eqnarray*}
and
\begin{eqnarray*}
	F_W(w) = \frac{1}{\Gamma(\theta)}\gamma(\theta, \tau w).
\end{eqnarray*}
We have
\begin{eqnarray*}
	E[T_m]&=&\frac{\int_0^\infty y\frac{\beta}{\alpha}\left(\frac{y}{\alpha}\right)^{\beta-1}e^{-\left(y/\alpha\right)^\beta} (1-(1-\prod_{j=1}^k \bigg[1-\frac{1}{\Gamma(m\eta_j)}\gamma(m\eta_j, \delta_jy)\bigg ]) )dy}{\int_0^{\infty}(1-\prod_{j=1}^k \bigg[1-\frac{1}{\Gamma(m\eta_j)}\gamma(m\eta_j, \delta_js)\bigg ])\frac{\beta}{\alpha}\left(\frac{y}{\alpha}\right)^{\beta-1}e^{-\left(s/\alpha\right)^\beta}ds}\\
	&&+\frac{(1-\int_0^{\infty}(1-\prod_{j=1}^k \bigg[1-\frac{1}{\Gamma(m\eta_j)}\gamma(m\eta_j, \delta_js)\bigg ])f_Y(s)ds)\int_0^\infty w\frac{\tau^\theta}{\Gamma(\theta)}w^{\theta-1}e^{-\tau w}dw}{\int_0^{\infty}(1-\prod_{j=1}^k \bigg[1-\frac{1}{\Gamma(m\eta_j)}\gamma(m\eta_j, \delta_js)\bigg ])\frac{\beta}{\alpha}\left(\frac{y}{\alpha}\right)^{\beta-1}e^{-\left(s/\alpha\right)^\beta}ds}\\
	&&+\frac{\int_0^\infty (-\prod_{j=1}^k \bigg[1-\frac{1}{\Gamma(m\eta_j)}\gamma(m\eta_j, \delta_js)\bigg ])'(1-(1-e^{-(s/\alpha)^\beta}))ds}{\int_0^{\infty}(1-\prod_{j=1}^k \bigg[1-\frac{1}{\Gamma(m\eta_j)}\gamma(m\eta_j, \delta_js)\bigg ])\frac{\beta}{\alpha}\left(\frac{y}{\alpha}\right)^{\beta-1}e^{-\left(s/\alpha\right)^\beta}ds}\\
		\end{eqnarray*}
\begin{eqnarray*}
	&=&\frac{\int_0^\infty y\frac{\beta}{\alpha}\left(\frac{y}{\alpha}\right)^{\beta-1}e^{-\left(y/\alpha\right)^\beta} \prod_{j=1}^k \bigg[1-\frac{1}{\Gamma(m\eta_j)}\gamma(m\eta_j, \delta_jy)\bigg ]dy}{\int_0^{\infty}(1-\prod_{j=1}^k \bigg[1-\frac{1}{\Gamma(m\eta_j)}\gamma(m\eta_j, \delta_js)\bigg ])\frac{\beta}{\alpha}\left(\frac{y}{\alpha}\right)^{\beta-1}e^{-\left(s/\alpha\right)^\beta}ds}\\
	&&+\frac{(1-\int_0^{\infty}(1-\prod_{j=1}^k \bigg[1-\frac{1}{\Gamma(m\eta_j)}\gamma(m\eta_j, \delta_js)\bigg ])f_Y(s)ds)\int_0^\infty w\frac{\tau^\theta}{\Gamma(\theta)}w^{\theta-1}e^{-\tau w}dw}{\int_0^{\infty}(1-\prod_{j=1}^k \bigg[1-\frac{1}{\Gamma(m\eta_j)}\gamma(m\eta_j, \delta_js)\bigg ])\frac{\beta}{\alpha}\left(\frac{y}{\alpha}\right)^{\beta-1}e^{-\left(s/\alpha\right)^\beta}ds}\\
	&&+\frac{\int_0^\infty (-\prod_{j=1}^k \bigg[1-\frac{1}{\Gamma(m\eta_j)}\gamma(m\eta_j, \delta_js)\bigg ])'e^{-(s/\alpha)^\beta}ds}{\int_0^{\infty}(1-\prod_{j=1}^k \bigg[1-\frac{1}{\Gamma(m\eta_j)}\gamma(m\eta_j, \delta_js)\bigg ])\frac{\beta}{\alpha}\left(\frac{y}{\alpha}\right)^{\beta-1}e^{-\left(s/\alpha\right)^\beta}ds}.
\end{eqnarray*}
Also, 
\begin{eqnarray*}
		P_{mk}(t)&=&1-(1-\prod_{j=1}^k \bigg[1-\frac{1}{\Gamma(m\eta_j)}\gamma(m\eta_j, \delta_jt)\bigg ])-(1-e^{-(t/\alpha)^\beta})\\
		&&+(1-\prod_{j=1}^k \bigg[1-\frac{1}{\Gamma(m\eta_j)}\gamma(m\eta_j, \delta_jt)\bigg ])(1-e^{-(t/\alpha)^\beta})\\
	&&+(\frac{1-e^{-(t/\alpha)^\beta}}{P(Y\leq Z_m)}1_{\{0\leq t \leq Z_m\}}+1_{\{t > Z_m\}}-\frac{P\{Y+W \leq t\}}{P\{Y\leq Z_m\}}1_{\{0\leq t \leq W+Z_m\}}-1_{\{t >W+Z_m\}})\\
	&&\times\int_0^{\infty}(1-e^{-(y/\alpha)^\beta})(-\prod_{j=1}^k \bigg[1-\frac{1}{\Gamma(m\eta_j)}\gamma(m\eta_j, \delta_jz)\bigg ])'dy\\
	&&+\int_0^t 1-(1-\prod_{j=1}^k \bigg[1-\frac{1}{\Gamma(m\eta_j)}\gamma(m\eta_j, \delta_j(t-s))\bigg ])-(1-e^{-((t-s)/\alpha)^\beta})\\
	&&+(1-\prod_{j=1}^k \bigg[1-\frac{1}{\Gamma(m\eta_j)}\gamma(m\eta_j, \delta_j(t-s))\bigg ])(1-e^{-((t-s)/\alpha)^\beta})dG(s)\\
	&&+\int_0^t(\frac{1-e^{-((t-s)/\alpha)^\beta}}{P(Y\leq Z_m)}1_{\{0\leq t-s \leq Z_m\}}+1_{\{t-s > Z_m\}}\\
	&&-\frac{P\{Y+W \leq t-s\}}{P\{Y\leq Z_m\}}1_{\{0\leq t-s \leq W+Z_m\}}-1_{\{t-s >W+Z_m\}})\\
	&&\times \int_0^{\infty}1-e^{-(y/\alpha)^\beta}(-\prod_{j=1}^k \bigg[1-\frac{1}{\Gamma(m\eta_j)}\gamma(m\eta_j, \delta_jz)\bigg ])'dydG(s)\\
		\end{eqnarray*}
\begin{eqnarray*}
	&=&\prod_{j=1}^k \bigg[1-\frac{1}{\Gamma(m\eta_j)}\gamma(m\eta_j, \delta_jt)\bigg ]-1+e^{-(t/\alpha)^\beta}\\
		&&+(1-\prod_{j=1}^k \bigg[1-\frac{1}{\Gamma(m\eta_j)}\gamma(m\eta_j, \delta_jt)\bigg ])(1-e^{-(t/\alpha)^\beta})\\
	&&+(\frac{1-e^{-(t/\alpha)^\beta}}{P(Y\leq Z_m)}1_{\{0\leq t \leq Z_m\}}+1_{\{t > Z_m\}}-\frac{P\{Y+W \leq t\}}{P\{Y\leq Z_m\}}1_{\{0\leq t \leq W+Z_m\}}-1_{\{t >W+Z_m\}})\\
	&&\times\int_0^{\infty}(1-e^{-(y/\alpha)^\beta})(-\prod_{j=1}^k \bigg[1-\frac{1}{\Gamma(m\eta_j)}\gamma(m\eta_j, \delta_jz)\bigg ])'dy\\
	&&+\int_0^t \prod_{j=1}^k \bigg[1-\frac{1}{\Gamma(m\eta_j)}\gamma(m\eta_j, \delta_j(t-s))\bigg ]-1+e^{-((t-s)/\alpha)^\beta}\\
	&&+(1-\prod_{j=1}^k \bigg[1-\frac{1}{\Gamma(m\eta_j)}\gamma(m\eta_j, \delta_j(t-s))\bigg ])(1-e^{-((t-s)/\alpha)^\beta})dG(s)\\
	&&+\int_0^t(\frac{1-e^{-((t-s)/\alpha)^\beta}}{P(Y\leq Z_m)}1_{\{0\leq t-s \leq Z_m\}}+1_{\{t-s > Z_m\}}\\
	&&-\frac{P\{Y+W \leq t-s\}}{P\{Y\leq Z_m\}}1_{\{0\leq t-s \leq W+Z_m\}}-1_{\{t-s >W+Z_m\}})\\
	&&\times \int_0^{\infty}(1-e^{-(y/\alpha)^\beta})(-\prod_{j=1}^k \bigg[1-\frac{1}{\Gamma(m\eta_j)}\gamma(m\eta_j, \delta_jz)\bigg ])'dydG(s).
\end{eqnarray*}

%
\end{document}